\documentclass[journal]{IEEEtran}

\usepackage{amsmath,amsfonts}
\usepackage{algorithmic}
\usepackage{algorithm}
\usepackage{array}
\usepackage[caption=false,font=normalsize,labelfont=sf,textfont=sf]{subfig}
\usepackage{textcomp}
\usepackage{stfloats}
\usepackage{url}
\usepackage{verbatim}
\usepackage{graphicx}
\usepackage{paracol}
\usepackage{cite}
\usepackage{framed}
\pdfoutput=1
\usepackage{amssymb}
\usepackage{xcolor}
\usepackage{amsthm}
\usepackage{bm}
\usepackage{hyperref}
\usepackage{scalerel}
\usepackage{float}
\usepackage{booktabs}

\newtheoremstyle{custom}
  {\topsep}    
  {\topsep}    
  {\normalfont} 
  {1em}           
  {\itshape}   
  {:}          
  { }          
  {\thmname{\textit{#1}}
   \thmnumber{ \textit{#2}}
   \thmnote{ \textit{(#3)}}}

\theoremstyle{custom}
\newtheorem{Thm}{Theorem}
\newtheorem{Lem}{Lemma}

\newtheorem{Prob}{Problem}
\newtheorem{subproblem}{Problem}[Prob]

\newcommand{\RNum}[1]{\uppercase\expandafter{\romannumeral #1\relax}}

\begin{document}
\title{A Unified QoS-Aware Multiplexing Framework for Next Generation Immersive Communication with Legacy Wireless Applications}
\author{Jihong Li,~\IEEEmembership{Student Member,~IEEE,} Shunqing Zhang,~\IEEEmembership{Senior Member,~IEEE,}~Tao Yu, 
~Guangjin Pan,~Kaixuan Huang,~\IEEEmembership{Student Member,~IEEE,}~Xiaojing Chen,~\IEEEmembership{Member,~IEEE,}~Yanzan Sun,~\IEEEmembership{Member,~IEEE,}~Junyu Liu,~\IEEEmembership{Member,~IEEE,}~Jiandong Li,~\IEEEmembership{Fellow,~IEEE,}~and Derrick Wing Kwan Ng,~\IEEEmembership{Fellow,~IEEE}
\thanks{This work was supported by the National Key Research and Development Program of China under Grants 2022YFB2902304, and the Science and Technology Commission Foundation of Shanghai under Grants 24DP1500500, and 24DP1500700. (Corresponding Author: {\em Shunqing Zhang}. )}
\thanks{Jihong Li, Shunqing Zhang, Tao Yu, Kaixuan Huang, Xiaojing Chen, and Yanzan Sun are with School of Communication and Information Engineering, Shanghai University, Shanghai, 200444, China (e-mails: \{tomlijiong, shunqing, yu\_tao, xuan1999, jodiechen, yanzansun\}@shu.edu.cn).

Guangjin Pan is with the Department of Electrical Engineering, Chalmers University of Technology, 41296 Gothenburg, Sweden (e-mail: guangjin.pan@chalmers.se). 

Junyu Liu and Jiandong Li are with State Key Laboratory of Integrated Service Networks, Institute of Information Science, Xidian University, Xi’an, Shaanxi, China (e-mails: jyliu@stu.xidian.edu.cn; jdli@mail.xidian.edu.cn). 

Derrick Wing Kwan Ng is with the School of Electrical Engineering and Telecommunications, University of New South Wales, Sydney, NSW, Australia (e-mail: w.k.ng@unsw.edu.au).}
}

\markboth{Journal of \LaTeX\ Class Files,~Vol.~14, No.~8, August~2021}%
{Shell \MakeLowercase{\textit{et al.}}: A Unified QoS-Aware Multiplexing Framework for Next Generation Immersive Communication with Legacy Wireless Applications}

\maketitle

\begin{abstract}
Immersive communication, including emerging augmented reality, virtual reality, and holographic telepresence, has been identified as a key service for enabling next-generation wireless applications. To align with legacy wireless applications, such as enhanced mobile broadband or ultra-reliable low-latency communication, network slicing has been widely adopted. However, attempting to statistically isolate the above types of wireless applications through different network slices may lead to throughput degradation and increased queue backlog. To address these challenges, we establish a unified QoS-aware framework that supports immersive communication and legacy wireless applications simultaneously. Based on the Lyapunov drift theorem, we transform the original long-term throughput maximization problem into an equivalent short-term throughput maximization weighted by virtual queue length. Moreover, to cope with the challenges introduced by the interaction between large-timescale network slicing and short-timescale resource allocation, we propose an adaptive adversarial slicing (Ad2S) scheme for networks with invarying channel statistics. To track the network channel variations, we also propose a measurement extrapolation-Kalman filter (ME-KF)-based method and refine our scheme into Ad2S-non-stationary refinement (Ad2S-NR). Through extended numerical examples, we demonstrate that our proposed schemes achieve 3.86 Mbps throughput improvement and 63.96\% latency reduction with 24.36\% convergence time reduction. Within our framework, the trade-off between total throughput and user service experience can be achieved by tuning systematic parameters.
\end{abstract}


\section{Introduction}
\IEEEPARstart{I}{mmersive} communication\cite{gachet_recommendation_2023}, also known as mobile broadband and low-latency (MBBLL) communications \cite{alwis_survey_2021}, is anticipated to deliver a real-time, interactive video experience in the coming decade, encompassing augmented reality (AR), virtual reality (VR), and holographic telepresence (HT). As an evolution of legacy enhanced mobile broadband (eMBB) and ultra-reliable low-latency communications (URLLC), immersive communication requires not only at the level of extrmemely high throughput with ten gigabits-per-second (Gbps), but also latency as low as few milliseconds, to accommodate the human visual system \cite{elbamby_toward_2018}. Recently, both the international telecommunication union (ITU) and the third generation partnership project (3GPP) have initiated the standardization process for immersive communication, which regarding it as an key enabler for real time interactive video experiences \cite{gachet_recommendation_2023} and launching a study item in the forming Release 19 \cite{3gpp_3gpp_2024}.

A one word approach to deploying immersive applications alongside legacy wireless applications is to exploit the radio access network (RAN) slicing scheme, which, allocates resources to different applications by simply multiplexing them into different RAN slices. In the past 5G era, several innovative resource multiplexing schemes have been proposed, including orthogonal multiplexing with statistical traffic requirements and wireless environments \cite{zhang_joint_2018}, non-orthogonal multiplexing through puncturing \cite{anand_joint_2020}, and superposition based scheme \cite{almekhlafi_superposition-based_2022}, or even combinations of these methods \cite{setayesh_resource_2022}. 
Despite these advancements, including quality of service (QoS) provisioning \cite{zhao_resource_2022} and predictive approaches \cite{chiariotti_temporal_2024}, existing techniques have proven inadequate for dealing with immersive applications as will be explained later. Consequently, developing a unified QoS-aware slice multiplexing framework is of utmost importance. 


\subsection{Prior Works}

During recent years, RAN slicing technologies have been extensively studied, which covers a wide range of orthogonal, non-orthogonal, and hybrid multiplexing schemes \cite{setayesh_resource_2022}. Specifically in the orthogonal multiplexing, different network slices are dedicated to serve different types of applications based on traffic and channel state statistics. For example, 
a sample approximation approach (SAA) based method was developed to operate long timescale wireless resource slicing with known channel statistics in \cite{tang_service_2019}, while the extensions to accommodate bursty traffic profiles and massive Internet of things (IoT) use cases have been explored in \cite{yang_multicast_2021, yang_how_2021, yang_ran_2021}, respectively. 
{\color{black}
On the other hand, in non-orthogonal multiplexing, puncturing is commonly adopted to transmit URLLC packets on top of eMBB services \cite{wang_resource_2023-1}, where the instantaneous throughput for eMBB service is maximized and the delay requirements for URLLC can be guaranteed. 
}
Another interesting non-orthogonal scheme is to exploit superposition-based transmission \cite{almekhlafi_superposition-based_2022}, where eMBB and URLLC services are superimposed together to satisfy their respective requirements simultaneously. Furthermore, by combining the above orthogonal and non-orthogonal multiplexing schemes, a hierarchical deep learning based method was proposed in \cite{setayesh_resource_2022}. Generally speaking, the aforementioned methods focus primarily on physical layer QoS assurance, while the higher layer QoS, such as packet backlog, have been largely overlooked.  

To address this issue, the integration of advanced higher layer queuing system models into existing network slicing (NS) framework has been pursued. 
{\color{black}
For instance, based on the $M/M/1$ queuing model, \cite{karbalaee_motalleb_resource_2023} illustrate different QoS profiles for eMBB, URLLC and massive machine type communications and provision adequate virtual network functions for each slices. 
Furthermore, in \cite{shi_two-level_2022}, the authors model the end-to-end delay for time-critical services in flying car based on queuing theory.  
On the other hand, an alternative approach involved applying the Lyapunov drift theorem \cite{kwak_dynamic_2017} to investigate the long-term stability issue of the current network slicing framework, where virtual queues are constructed to guarantee long-term latency requirements of services \cite{tu_deep_2024,zhao_resource_2022}. 
}
{\color{black}
Moreover, for practical scenarios which the short-term performance is more preferable, Markov decision process (MDP) based methods, such as reinforcement learning (RL) \cite{zhou_performance_2024}, have been widely adopted to address workload management problems in the RAN slicing framework. 
}
In addition, to balance the higher layer and the physical layer performance, a hierarchical structure employing multi-timescale MDP was proposed in \cite{mei_intelligent_2021} to optimize packet dropping rates and spectral efficiency simultaneously. All the above methods, however, focus on isolated QoS requirements with static channel environments, and the channel non-stationary issues with coupled QoS requirements for the upcoming 6G RAN scenarios are yet to be solved, highlighting the need for further research. 

To deal with non-stationary channel variations in next-generation immersive communication systems, predictive modeling has been widely studied for channel estimation, trajectory tracking, and other related areas. For example, autoregressive parameter estimation \cite{vinogradova_estimating_2022} and machine learning \cite{shi_unified_2021}-based channel estimation schemes were proposed to track theses non-stationary channel variations. 
{\color{black}
Besides, long short-term memory and echo state networks are incorporated into reinforcement learning frameworks, aiming at tracking environment variations \cite{choi_deep_2024, yang_feeling_2022}. 
Furthermore, Markov chain-assisted and context-aware online learning was leveraged in \cite{zanzi_laco_2021, zhang_learning_2024} to enhance the network slicing performance for same purposes. 
}
The coupled QoS requirement issue was discussed in \cite{chaccour_can_2022}, where a novel and rigorous characterization of the end-to-end delay distribution for high throughput transmission was identified as crucial. Moreover, in \cite{zhao_online_2024}, these coupled QoS requirements were incorporated into nonlinear knapsack problems and efficiently handled by MDP based online scheduling policies. 
{\color{black}
In fact, there is still limited works to deal with the non-stationary and the coupled QoS requirement issues simultaneously, and the straightforward approach to combine the aforementioned solutions may fail to work due to the non-stochastic multi-timescale interactions and the parameter catastrophic forgetting issues, as explained in \cite{zhang_learning_2024} and \cite{bing_meta-reinforcement_2023}. To address the above challenges, in this paper, we provide non-stationary and non-stochastic-robust slicing schemes under a unified legacy/6G QoS oriented framework.  
}

\subsection{Contributions}
{\color{black}
The main contribution points of our work are as follows:  

\begin{itemize}
    \item {{\bf Hybrid legacy / 6G QoS oriented framework: }
    To achieve a smooth transition to 6G, a joint framework that characterizes heterogeneous services is challenging and non-trivial. Specifically, to the best of our knowledge, for the legacy services, such as eMBB, URLLC as well as the 6G MBBLL services, a unified QoS optimization framework that distinguishes between both strict and loose delay requirements, different service intensity as well as hard and soft QoS fulfillment is not yet investigated. In this paper, we model distinct traffic intensities and MAC-layer backlog thresholds for mixed QoS requirements, incorporated by Lyapunov-based backlog-aware resource allocation to reduce instantaneous delay outage. 
    }
    \item {{\bf Multi-timescale slicing approach with non-stochastic-robustness: }
    Managing complicated multi-timescale interactions in the next-generation network slicing framework presents significant challenges. 
    To be specific, slicing orchestrators in the existing works fail to learn efficiently and provide optimal decisions due to their inability to capture system variation, and deal with the non-stochastic outcomes of frame level resource allocation policies. 
    In this paper, backlog-aware design and non-stochastic-robustness are introduced in our network slicing approach to develop online adaptive adversarial slicing (Ad2S) algorithm with theoretical performance guarantees. This method has determined effectiveness, achieving 
    a less than 73\% cumulative regret against benchmark algorithms, as validated through extensive experiments. 
    }
    \item {{\bf Refined slicing mechanism with non-stationary channel tracking: } 
    Non-stationary channel variations disrupt the asymptotical optimality of Ad2S algorithm due to parameter forgetting issues. Simply applying AI-based channel tracking schemes results in growing complexity and lack of theoretical interpretability. 
    In this paper, we propose a low-complexity context enhancement scheme that incorporates results from an unbiased measurement extrapolation-Kalman filter (ME-KF) channel tracking scheme, thus effectively eliminating the forgetting issues under temporal non-stationary channel. 
    Our theoretical analysis demonstrates that our joint design achieves sub-linear regret, even in non-stationary scenarios. 
    }
\end{itemize}

For the validation of the up-mentioned contributions, we carry out extensive tests of metrics including cumulative regret, average transmission rate, experienced latency, and task satisfaction rate. With our proposed schemes, 3.86 Mbps average throughput improvement, a 63.96\% reduction in latency and a 24.36\% reduction in convergence time are achieved. 
}

The remainder of this paper is organized as follows. In Section~\ref{system model}, we discuss the mathematical models of the unified QoS-aware multiplexing framework for immersive 6G communication with legacy wireless applications and formulate the problem of maximization of average throughput in Section~\ref{problem formulation}. Based on the hierarchical optimization framework, we propose a relaxation-based resource allocation scheme and context-aware Ad2S network slicing scheme for scenarios with stationary channel statistics in Section~\ref{sta} and refine the latter with ME-KF based channel tracker under non-stationary cases in Section~\ref{non-sta}. In Section~\ref{performance analysis} and Section~\ref{simulation results}, analytical and numerical results are presented to demonstrate the performance improvement of the proposed schemes. Finally, concluding remarks are given in Section \ref{conclusion}.

\section{System Model}\label{system model}

In this section, we introduce the network topology and the adopted wireless propagation models, followed by some traffic and queuing dynamics for different applications.
\begin{figure}[!t]
\centering
\subfloat[]{
		\includegraphics[scale=0.28]{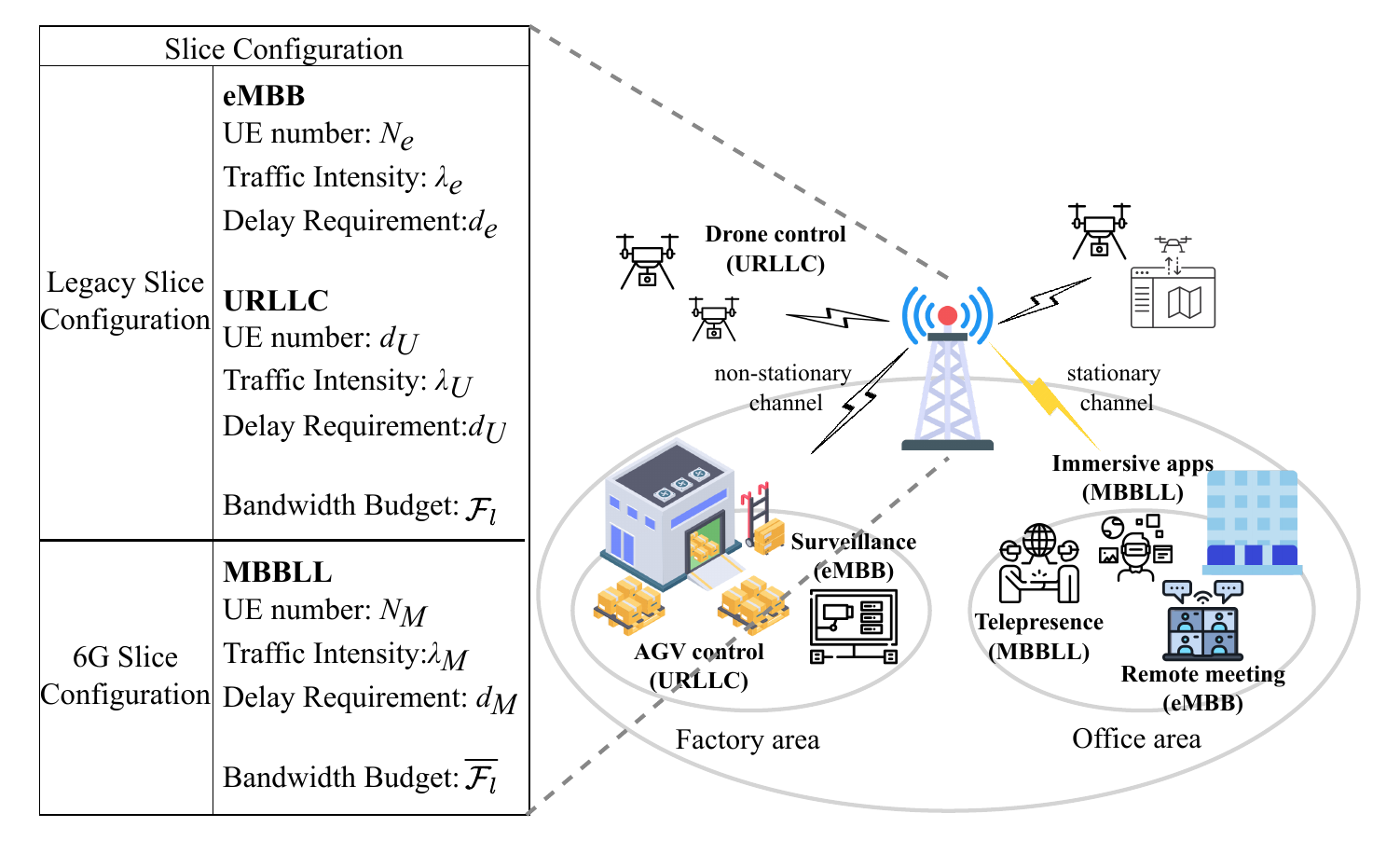}}\\
\subfloat[]{
		\includegraphics[scale=0.2, trim=30 0 10 20, clip]{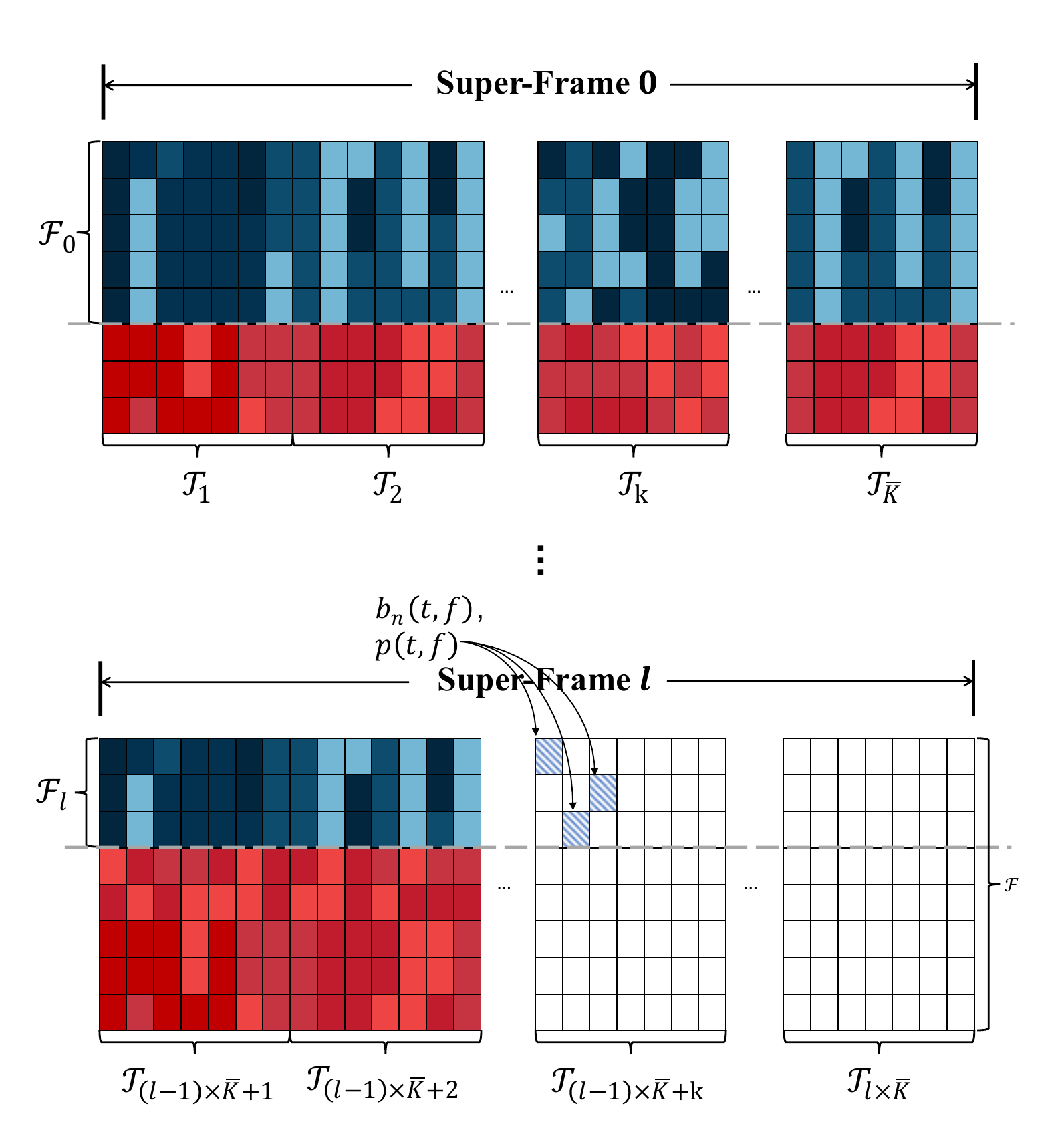}}
    \caption{(a) System model; (b) Dual-scale resource grid for network slicing and resource allocation. }
\vspace{-3mm}
\label{fig:sys_mod}
\end{figure}

\subsection{Topology and Propagation Models}
Consider an Orthogonal Frequency-Division Multiple Access (OFDMA) network with a single base station (BS) covering a circular area with radius $r$, $N_{e}$ eMBB users with indexes $\mathcal{N}_e \in \{1,\ldots, N_e\}$, $N_{U}$ URLLC users with indexes $ \mathcal{N}_U \in \{ N_e+1,\ldots, N_e + N_{U}\}$, and $N_{M}$ MBBLL users with indexes $\mathcal{N}_M \in \{ N_e + N_{U}+1,\ldots, N_e + N_{U} + N_{M}\}$, as shown in Fig.~\ref{fig:sys_mod}. 
As illustrated in Fig.~\ref{fig:sys_mod}(a), the hybrid 5G/6G oriented downlink transmission cell system serves users under a multiplexing scenario with co-existing eMBB, URLLC, MBBLL services. Also, each service is characterized by the number of users $\mathcal{N}_i$, service intensity $\lambda_i$, and latency requirement $d_i$, where $i \in \{e, U, M\}$. The base station transmits data in the downlink according to time-frequency resource grid $(t, f) \in \mathcal{T} \times \mathcal{F}$, which follows a frame and super-frame-based dual-timescale structure as shown in Fig.~\ref{fig:sys_mod}(b) \cite{zhang_joint_2018, tang_service_2019, yang_how_2021}. A sum of $\Bar{K}$ frames are contained in a super-frame. At the start of super-frame $l$, $l \in \{1, \ldots, L\}$, the resource grid is sliced according to some specific network slicing strategy. At the start of frame $k$, $k \in \{1, \ldots, \Bar{K}\}$, the legacy eMBB and URLLC (illustrated with blue color) and 6G MBBLL (illustrated with red color) services are allocated to resource grids $(t, f) \in \mathcal{T}_k \times \mathcal{F}$.

For any given time slot $t$ and sub-channel $f$, the equivalent transmission capacity for the $n^{\mathrm{th}}$ user is given by,
\begin{eqnarray} \label{eqn:capa}
r_n(t,f) =  B \log_2 \left(1 + b_n(t,f) |h_n(t,f)|^2  p(t,f)\right),
\end{eqnarray}
where $B$ denotes the sub-channel bandwidth, $h_n(t,f)$ denotes the normalized complex-valued channel fading coefficient\footnote{In this paper, we normalize the channel coefficient with respect to the noise power for illustration purposes.} between the BS and the $n^{\mathrm{th}}$ user, and $b_n(t,f) \in\{0,1\}$ denotes the indicator function, which equals to $1$ if the time-frequency element $(t,f)$ is allocated to the $n^{\mathrm{th}}$ user and zero otherwise. Since each time-frequency element can only be allocated executively to one user within a time frame $\mathcal{T}_k \times \mathcal{F}$, we have
\begin{eqnarray}
\label{constraint sub-channel}
\sum_{n \in \mathcal{N}_e \cup \mathcal{N}_U \cup \mathcal{N}_M}b_n(t,f) \leq 1, &&\forall(t,f) \in \mathcal{T}_k \times \mathcal{F}.
\end{eqnarray}
Besides, $p(t,f)$ denotes the allocated transmission power for each time-frequency element $(t,f)$, which satisfies
\begin{eqnarray}
\label{constraint power}
\sum_{f \in \mathcal{F}} 
p(t,f) \leq P_{\mathrm{tot}}.
\end{eqnarray}  
In the above formulation, for time frame  $k$ with index $\mathcal{T}_k$ and bandwidth $\mathcal{F}$, the transmission throughput of the $n^{\mathrm{th}}$ user, $r_n(k)$, is thus given by
\begin{eqnarray}
r_n(k) = \sum_{(t,f) \in \mathcal{T}_k \times \mathcal{F}} r_n(t,f).
\label{frame_capa}\end{eqnarray}

\subsection{Traffic Models}
Denote $Q_n(k)$ to be the queue status of the $k^{\mathrm{th}}$ time frame for the $n^{\mathrm{th}}$ user. Then, the queuing dynamic is 
\begin{eqnarray}
\label{dynamic queue model}
    Q_{n}(k+1) = \max\{Q_{n}(k) - \eta r_n(k),0\} + \Lambda_n(k),
\end{eqnarray}
where $\eta = \frac{K_f}{L_{\mathrm{packet}}}$ denotes the normalized parameter corresponding to frame duration $K_f$ and packet length $L_{\mathrm{packet}}$, $\Lambda_n(k)$ is the number of arrived packets during the $k^{\mathrm{th}}$ time frame with 
\begin{eqnarray}
\mathbb{E}[\Lambda_n(k)] = 
\left\{
\begin{array}{l l}
  \lambda_e, & n \in \mathcal{N}_e,\\
  \lambda_U, & n \in \mathcal{N}_U, \\
  \lambda_M, & n \in \mathcal{N}_M,
\end{array}
\right.
\label{arrival}
\end{eqnarray}
where $\mathbb{E}(\cdot)$ denotes the mathematical expectation and $\lambda_e, \lambda_U, \lambda_M$ denote the average arrival rates for eMBB, URLLC, and MBBLL users, respectively. Once we denote $K = L \times \Bar{K}$ to be the total number of time frames, the average queuing backlogs for eMBB, URLLC, MBBLL UEs are given by, 
\begin{eqnarray}
\label{constraint average queue length}
\begin{aligned}
    &\limsup_{K\to\infty}\frac{1}{K}\sum_{k=0}^{K-1}Q_{n}(k)\leq 
    \left\{
    \begin{array}{l l}
      \delta_e, & n \in \mathcal{N}_e,\\
      \delta_U, & n \in \mathcal{N}_U, \\
      \delta_M, & n \in \mathcal{N}_M,
    \end{array} \right.
\end{aligned}
\end{eqnarray}
where $\delta_e = \lambda_e d_e$, $\delta_U = \lambda_U d_U$, and $\delta_M = \lambda_M d_M$ denote the delay requirements for eMBB, URLLC, and MBBLL users, respectively, as governed by the celebrating Little's law in the area of queuing theory \cite{little_littles_2008}. On the other hand, according to the 3GPP TR 38.913 \cite{noauthor_study_2022}, the URLLC packets should be delivered within a frame duration to satisfy the stringent delay requirement such that the above constraint \eqref{constraint average queue length} becomes
\begin{eqnarray}
&\eta r_n(k) \geq Q_n(k), n \in \mathcal{N}_U, \label{constraint URLLC rate} \\
&\underset{K\to\infty}{\limsup}\frac{1}{K}\displaystyle\sum_{k=0}^{K-1}Q_{n}(k)\leq 
\left\{
\begin{array}{l l}
  \delta_e, & n \in \mathcal{N}_e,\\
  \delta_M, & n \in \mathcal{N}_M.
\end{array} \right. \label{constraint average queue length, no URLLC}
\end{eqnarray}

{\color{black}
Furthermore, to isolate the legacy 5G eMBB and URLLC with the 6G MBBLL services, a frequency domain network slicing strategy has been adopted with the time-varying partition $\mathcal{F}_l$ and $\overline{\mathcal{F}_l}$\footnote{{\color{black}This isolation strategy aligns with 3GPP’s network slicing principles \cite{3GPP_38_401, 3GPP_28_530}, while addresses multi-generation coexistence challenges and provides QoS guarantees.}}, such that for the $l^{\mathrm{th}}$ super-frame, we have
}
\begin{eqnarray}
b_n(t,f) = 0, && \forall n \in \mathcal{N}_e \cup \mathcal{N}_U \  \& \ f \in \overline{\mathcal{F}_l}, \label{eqn:f_l_legacy}\\
b_n(t,f) = 0, && \forall n \in \mathcal{N}_M \ \& \ f \in \mathcal{F}_l, \label{eqn:f_l_6g}\\
\mathcal{F}_l \cup \overline{\mathcal{F}_l} = \mathcal{F}. \label{eqn:f_l}
\end{eqnarray}

The following assumptions are adopted through the rest of this paper. 
First, we consider the standard log-normal shadowing, where the magnitude of channel coefficients, i.e., $|h_n(t,f)|$, follows the super-frame based log-normal distribution with mean $\mu_{n}(l)$ and variance $\sigma^2_{n}$, where $\mu_{n}(l)$ stand for the pathloss component. 
Second, the temporal evolution of the pathloss component follows an autoregression model as given by \cite{baddour_autoregressive_2005, chopra_adaptive_2024}, 

\begin{small}
\begin{eqnarray}
    \begin{bmatrix}
        \mu_{1}(l)\\
        \ldots\\
        \mu_{N}(l)
    \end{bmatrix} 
    = \bm{A}(l - 1)
    \begin{bmatrix}
       \mu_{1}(l-1)\\
        \ldots\\
       \mu_{N}(l-1)
    \end{bmatrix} + \bm{N}_{\mu}, \quad l \geq 1. 
\label{pathloss_ar}
\end{eqnarray}
\end{small}

\noindent In the above expression, $\bm{A}(l) = \textrm{diag}[a_1(l), \ldots, a_{N}(l)]$ denotes the temporal correlation of the $l^{th}$ super-frame that evolves according to $\bm{A}(l) = \beta \bm{A}(l - 1) + \bm{N}_{\bm{A}}$, where $\textrm{diag}(\bm{a})$ denotes the diagonal matrix composed with vector $\bm{a}$, and $\beta$ is any positive constant that satisfies $0 \leq \beta < 1$ which stands for the damping term of classical autoregression model as specified in \cite{hamilton_time_1994}. 
Third, complex-valued channel coefficients $h_n(t, f)$ are sensed in each allocation frame. 
Fourth, $\bm{N}_{\mu}$ and $\bm{N}_{\bm{A}}$ follow the standard zero mean Gaussian distributions with covariance matrices $\bm{Q}_{\mu}$ and $\bm{Q}_{\bm{A}}$, respectively. 
Fifth, $\Lambda_n(k)$ follows Poisson distribution \cite{karagiannis_nonstationary_2004} with mean arrival rate $\lambda_n$. 
Last but not least, once we assume the minimum ergodic period of both $\{|h_n(t, f)|\}$ and $\{\lambda_n(k)\}$ as $K_e$, as specified in \cite{tang_service_2019, yang_how_2021}, we have $\Bar{K} > K_e$, and $\Bar{K}$ is sufficiently large such that $\{|h_n(t, f)|\}$ within each super-frame is independent and identically distributed (i.i.d.) and ergodic, which results from the significant difference in the timescales. 

\section{Problem Formulation and Transformation}\label{problem formulation}
In this section, we first formulate the average throughput maximization problem exploiting a general optimization framework and than discuss the Lyapunov drift-based formulation. 

To start with, based on the per-user throughput expression as given in \eqref{frame_capa}, we can calculate the total throughput of BS by accumulating over all the possible users, e.g., $n \in \mathcal{N}_e \cup \mathcal{N}_U \cup \mathcal{N}_M$. By averaging over the total number of time frames, $K$, the original throughput maximization problem can be formulated as follows. 
\begin{Prob}[Original Problem] The average throughput of the unified QoS-aware multiplexing framework can be maximized via the following optimization problem.
\begin{eqnarray}\label{prob_original}
    \underset{\{b_n(t,f)\}, \{p(t,f)\}, \{\mathcal{F}_l\}}{\textrm{maximize}} && \lim_{K \to \infty}\frac{1}{K}\sum_{k = 1}^{K}\sum_{n \in \mathcal{N}_e \cup \mathcal{N}_U \cup \mathcal{N}_M} r_n(k), \nonumber \\
    \textrm{subject to} \qquad
    && \eqref{constraint sub-channel}, \eqref{constraint power}, \eqref{constraint URLLC rate}-\eqref{eqn:f_l}, \nonumber
\end{eqnarray}
where the frequency domain network slicing strategy $\{\mathcal{F}_l\}$ is updated on a super-frame basis and other resources, e.g., $\{b_n(t,f)\}, \{p(t,f)\}$, are determined on a per time frame basis.
\label{prob URLLC}
\end{Prob}

By mixing the discrete network slicing strategy $\{\mathcal{F}_l\}$, the discrete bandwidth allocation $\{b_n(t, f)\}$, and the continuous power allocation $\{p(t, f)\}$, the above problem belongs to a NP-hard mixed integer non-linear programming (MINLP) problem as specified in \cite{boukouvala_global_2016}. In addition, with the long-term average queuing constraint \eqref{constraint average queue length, no URLLC} and coupled queuing dynamics due to shared wireless resources, Problem~\ref{prob URLLC} becomes even more challenging to solve in general. 

To make it mathematically tractable, 
we denote a virtual queue defined as $\Gamma(k)\triangleq[G_1(k), G_2(k), \ldots, G_n(k)]$ with dynamics $G_n(k + 1) = \max\{G_n(k) + Q_n(k + 1) - \delta_{e,M}, 0\}, n \in \mathcal{N}_e \cup \mathcal{N}_M$, and the corresponding weighted-Lyapunov function \cite{neely_mj_stochastic_2010} is given by,\begin{eqnarray}
    L(\Gamma(k))\triangleq\frac{\omega_{Q}}{2}\sum_{n \in \mathcal{N}_e \cup \mathcal{N}_M}{G_n(k)}^2,
\end{eqnarray}
where $\omega_{Q}>0$ is a weight factor for the virtual queue. By utilizing the Lyapunov drift theorem \cite{neely_mj_stochastic_2010}, we can solve Problem~\ref{prob URLLC} through minimizing the drift-minus-utility (DMU) function as summarized in the following lemma.

\begin{Lem}[Lyapunov Drift \cite{neely_mj_stochastic_2010}] \label{lem_dmu_long_term}
If we denote the DMU function to be $\Delta(\Gamma(k)) - \omega_{T} \mathbb{E}\left[\sum_{n \in \mathcal{N}_e \cup \mathcal{N}_U \cup \mathcal{N}_M} r_n(k)|\Gamma(k)\right]$, then the objective function in Problem~\ref{prob URLLC} can be maximized when the following upper bound of DMU function is minimized.

\begin{small}
\vspace{-6mm}
    \begin{align} \label{DMU upper bound}
        \Delta(\Gamma(k)) - \omega_{T}\mathbb{E}\left[ \sum_{n \in \mathcal{N}_e \cup \mathcal{N}_U \cup \mathcal{N}_M} r_n(k)|\Gamma(k) \right] \nonumber \\
        \leq \mathbb{E} \left[ \frac{\omega_{Q}}{2} \sum_{n \in \mathcal{N}_e \cup \mathcal{N}_M}\big(C_n(k) - 2G_n(k) \eta r_n(k)\big)|\Gamma(k)\right]\nonumber\\
        - \omega_{T}\mathbb{E}\left[ \sum_{n \in \mathcal{N}_e \cup \mathcal{N}_U \cup \mathcal{N}_M} r_n(k)|\Gamma(k) \right], \qquad \qquad \qquad \quad
    \end{align}
\end{small}

\noindent where $\Delta(\Gamma(k)) \triangleq \mathbb{E}\left[L(\Gamma(k + 1)) - L(\Gamma(k))|\Gamma(k)\right]$, $\omega_{T}>0$ denotes the weight factor for total throughput, and $C_n(k) = {(Q_n(k) + \Lambda_n(k))}^2 + {\delta_{e,M}}^2 + 2 G_n(k) Q_n(k) + 2 G_n(k) \Lambda_n(k)$ denotes the throughput irrelevant terms.
\end{Lem}
\begin{proof}
Please refer to Appendix~\ref{appendix:lem_1} for the proof.
\end{proof}
By applying Lemma~\ref{lem_dmu_long_term}, the original average throughput maximization problem can be equivalently transformed to the following drift-based formulation. 
\begin{Prob}[Drift-based Formulation] The upper bound of the DMU function can be optimized via the following optimization framework:

\begin{small}
       \begin{align}
        \underset{\{\mathcal{F}_l\}}{\textrm{max}} \quad & \lim_{L \to \infty} \frac{1}{L} \sum_{l = 1}^{L} \left\{\frac{1}{\Bar{K}}\sum_{k = (l - 1) \times \Bar{K} + 1}^{l \times \Bar{K}} F^*_k(\mathcal{F}_l)\right\}, \nonumber\\
        \textrm{s.t.} \quad & F^*_k(\mathcal{F}_l) \triangleq \begin{cases}
            \underset{\{b_n(t, f)\}, \atop \{p(t, f)\}}{\textrm{max}}& \displaystyle\sum_{n \in \mathcal{N}_e \cup \mathcal{N}_M} \omega_Q G_n(k) r_n(k) \nonumber\\
            & - \displaystyle\sum_{n \in \mathcal{N}_e \cup \mathcal{N}_M} \frac{\omega_Q}{2} C_n(k)\nonumber\\
            & + \displaystyle\sum_{n \in \mathcal{N}_e \cup \mathcal{N}_U \cup \mathcal{N}_M} \omega_T r_n(k), \nonumber\\
            \quad \textrm{s.t.} &
            \eqref{constraint sub-channel}, \eqref{constraint power}, \eqref{constraint URLLC rate}, \eqref{eqn:f_l_legacy}-\eqref{eqn:f_l}. \nonumber\\
            \end{cases} \\
        \nonumber
    \vspace{-6mm}
    \end{align} 
\end{small}
\label{prob Lyapunov}
\vspace{-6mm}
\end{Prob}
\noindent With the above transformation, we can effectively decouple the queueing dynamics into dual-timescale hierarchical optimization structure, where the inner optimization focuses on solving the short term power and bandwidth allocation within each time frame and the outer one solves the optimal network slicing strategy based on the obtained utility $F^*_k(\mathcal{F}_l)$.

\section{Dual-timescale Resource Allocation Under Stationary Case}\label{sta}

In this section, we focus on solving the drift-based Problem~\ref{prob Lyapunov} under the stationary condition, e.g., $\bm{A}(l) = \bm{I}, \bm{N}_{\bm{A}} = \bm{N}_{\mu} = \bm{0}$, and constraint \eqref{pathloss_ar} is degenerated into the following expression,  
\begin{align}
        [\mu_{1}(l) \ldots \mu_{N}(l)]^{T} = [\mu_{1}(l-1) \ldots \mu_{N}(l-1)]^{T}, l \geq 1. \label{stationary assumption}
\end{align}
With the dual-timescale hierarchical optimization framework as described in Problem~\ref{prob Lyapunov}, we decompose it into a frame-scale resource allocation problem for the fixed network slicing strategy $\mathcal{F}_l$, and then solve the optimal network slicing strategy in the super-frame scale as follows.
\begin{Prob} The optimal resource allocation strategy can be obtained by iteratively solving the following two subproblems.
\begin{subproblem}[Frame-Scale Resource Allocation]
    By dropping those irrelevant terms from $F_k^*(\mathcal{F}_l)$ such as $C_n(k)$, the frame scale resource allocation problem under any given $\mathcal{F}_l$ is equivalent to, 

    \begin{small}
    \vspace{-3mm}
    \begin{eqnarray}
         \underset{\{b_n(t, f)\}, \{p(t, f)\}}{\textrm{maximize}} && \displaystyle\sum_{n\in\mathcal{N}_e \cup \mathcal{N}_M} \omega_Q G_n(k) \eta r_n(k) \nonumber\\
         && + \sum_{n \in \mathcal{N}_e \cup \mathcal{N}_U \cup \mathcal{N}_M}\omega_T r_n(k)\label{inner_opt}\\
         \textrm{subject to} && \eqref{constraint sub-channel}, \eqref{constraint power}, \eqref{constraint URLLC rate}, \eqref{eqn:f_l_legacy}-\eqref{eqn:f_l}.  \nonumber
         \vspace{-6mm}
    \end{eqnarray} \label{prob frame}
    \end{small}
    
    \noindent Note that solving the frame-scale resource allocation problem is still challenging due to its MINLP and coupled nature for $\{b_n(t, f)\}$ and $\{p(t, f)\}$. 
\end{subproblem}

\begin{subproblem}[Super-frame Scale Network Slicing] With the obtained utility $F^*_k(\mathcal{F}_l)$,  the super-frame scale network slicing problem to maximize the overall utility is given by

\begin{small}
\vspace{-3mm}
\begin{eqnarray}
    \underset{\{\mathcal{F}_l\}}{\textrm{maximize}} && \lim_{L \to \infty} \frac{1}{L} \sum_{l = 1}^{L}\left\{\frac{1}{\Bar{K}}\sum_{k = (l - 1) \times \Bar{K} + 1}^{l \times \Bar{K}}  F^*_k(\mathcal{F}_l) \right\}. \label{obj_hierarchical stochastic problem}
\end{eqnarray}
\end{small}

The super-frame scale network slicing problem is also challenging as the objective function is non-convex and temporal correlated. 
To solve this, several methods have been proposed. 
Indeed, when a closed-form $F_k^*(\mathcal{F}_l)$ is available, the problem can be handled via standard dynamic programming methods \cite{r_bellman_dynamic_1966}. 
On the other hand, by introducing stochastic characteristic into $F_k^*(\mathcal{F}_l)$, efficient upper confidence bound (UCB)-based exploration schemes are proved to achieve asymptotical optimality under certain sub-Gaussian stochastic distributions \cite{zanzi_laco_2021}. 
{\color{black}
However, $F_k^*(\mathcal{F}_l)$ in the considered setting is a non-stochastic \cite{auer_nonstochastic_2002} objective function given under both the backlog status and the frame scale resource allocation policies, thus aforementioned methods may fail to exhibit optimal behavior. 
}

\label{hierarchical stochastic problem}
\end{subproblem}
\end{Prob}

\subsection{Frame Scale Resource Allocation}\label{lower-level_RA}
{\color{black}
\begin{algorithm}
    \newcommand{\algorithmicpower}{\textbf{Power Allocation Stage:}}
	\newcommand{\algorithmicsubchannel}{\textbf{Sub-channel Matching Stage:}}
        \caption{Proposed Algorithm for Problem~\ref{prob frame}}
    \label{alg_psum_bcd}
    \begin{algorithmic}[1]
    \STATE PSUM initialize: $\alpha > 1$, $\sigma^{(1)}$, let the number of outer iterations $n^{out} = 0$. 
    \WHILE{$n^{\mathrm{out}}<n^{\mathrm{out}}_{\mathrm{max}}$:}
        \STATE BCD initialize: $b_n(t, f)^{0} = \frac{1}{|\mathcal{N}|}, p(t, f)^{0} = \frac{P_{\mathrm{tot}}}{|\mathcal{T}_k||\mathcal{F}|}$, let the number of inner iterations $n^{\mathrm{in}} = 0$.
        \WHILE{$n^{\mathrm{in}}<n^{\mathrm{in}}_{max}$:}
            
            \STATE Solve $g_{n^{\mathrm{out}}}(\{b_n(t, f)\}^{n^{\mathrm{in}} + 1}, \{p(t, f)\}^{n^{\mathrm{in}} + 1})$ with CVX \cite{boyd_convex_2023} and acquire $\{b_n(t, f)\}^{n^{\mathrm{in}} + 1}$;\\
               
            \STATE Calculate the objective function of \eqref{inner_opt} which follows $g(\{b_n(t, f)\}^{n^{\mathrm{in}} + 1}, \{p(t, f)\}^{n^{\mathrm{in}} + 1}) - g(\{b_n(t, f)\}^{n^{\mathrm{in}}}, \{p(t, f)\}^{n^{\mathrm{in}}}) = \Delta_{n^{\mathrm{in}}}$; 
            \IF{$\Delta_{n^{\mathrm{in}}} \leq \Delta^{\mathrm{in}}$:}
                \STATE Output $\{b_n(t, f)\}^{n^{\mathrm{in}} + 1}, \{p(t, f)\}^{n^{\mathrm{in}} + 1}$;
            \ELSE
                \STATE Set $n^{\mathrm{in}} = n^{\mathrm{in}} + 1$;
            \ENDIF
        \ENDWHILE
        \STATE $\{b_n(t, f)\}^{n^{\mathrm{out}} + 1} = \{b_n(t, f)\}^{n^{\mathrm{in}} + 1}$; 
        \STATE $\{p(t, f)\}^{n^{\mathrm{out}} + 1} = \{p(t, f)\}^{n^{\mathrm{in}} + 1}$; 
        \IF{$\{b_n(t, f)\}^{n^{\mathrm{out}} + 1}$ is binary:}
            \STATE Output $\{b_n(t, f)\}^{(n^{\mathrm{out}} + 1)}, \{p(t, f)\}^{(n^{\mathrm{out}} + 1)}$;
        \ELSE
            \STATE Set $n^{\mathrm{out}} = n^{\mathrm{out}} + 1$, and let $\sigma^{(n^{\mathrm{out}} + 1)} = \alpha \sigma^{n^{\mathrm{out}}}$;\\
        \ENDIF
    \ENDWHILE
    \end{algorithmic}
\end{algorithm}
}

To solve Problem~\ref{prob frame}, we apply the penalty successive upper bound minimization (PSUM) approach as explained in \cite{zhang_network_2017}. Specifically, we relax the binary variables $\{b_n(t, f)\}$ to be continuous and augment the binary constraint via the $p$-order penalty function to the objective function for optimizing the lower bound of the original problem, block coordinate descent (BCD) is then utilized to obtain resource allocation solution with alternating iterations, this joint algorithm is renamed after penalty BCD-based resource allocation (PBRA) algorithm.  Detailed algorithms are summarized in Algorithm~\ref{alg_psum_bcd} and the convergence properties are discussed in Lemma~\ref{lem_bcd_converge}. 

{\color{black}
\begin{Lem}[Convergence of Alg.~\ref{alg_psum_bcd}]
    Within Alg.~\ref{alg_psum_bcd}, as $\sigma^{n^\mathrm{out}} \to +\infty$, $\{0, 1\}$-valued $\{b_n(t, f)\}$ and critical point $\{b_n(t, f), p(t, f)\}$ for Problem~\ref{prob frame} are achieved. 
    \begin{proof}
        Please refer to Appendix~\ref{appendix:lem_3_real} for the proof. 
    \end{proof}
    \label{lem_bcd_converge}
\end{Lem}
}
\subsection{Super-frame Scale Network Slicing}\label{subsec:saa admm}
With the obtained utility for given $\mathcal{F}_l$ through Algorithm~\ref{alg_psum_bcd}, e.g., $F^*_k(\mathcal{F}_l)$, we focus on solving the long-term utility maximization problem across different network slicing strategies $\{\mathcal{F}_l\}$. To address Problem~\ref{hierarchical stochastic problem}, we apply non-stochastic contextual bandit based methods\cite{neu_efficient_2020}, where the virtual and buffer queue statuses, e.g., $\{G_n((l - 1)\times\Bar{K}), Q_n((l - 1)\times\Bar{K})\}$, are chosen to be contextual features. Since the closed-form expression for $F^*_k(\mathcal{F}_l)$ is unavailable, we directly propose a network slicing algorithm as follows, namely Ad2S algorithm, and discuss the asymptotical optimality properties in Section~\ref{performance analysis}. 

\begin{algorithm}
        \caption{Proposed Algorithm for Problem~\ref{hierarchical stochastic problem}}
    \label{stationary_alg}
    \begin{algorithmic}[1]
    \STATE Initialization: Initialize learning rate, exploration parameter, and per-subcarrier loss, e.g., $\eta = L^{-\frac{2}{3}}(|\mathcal{F}|(1 + 4N))^{-\frac{1}{3}}(\log |\mathcal{F}|)^{\frac{2}{3}}$,$\gamma = L^{-\frac{1}{3}}(|\mathcal{F}|(1 + 4N)\log |\mathcal{F}|)^{\frac{1}{3}}$, $G_n(0) = Q_n(0) = 1, \forall n \in [1,N]$, $\pi_0(f|\bm{X}_l) = \frac{1}{|\mathcal{F}|}, \hat{\bm{\theta}}_{0, f} = \bm{0} \in \mathbb{R}^{1 + 4N}, \forall f \in [1, \ldots, |\mathcal{F}|]$;
    \FOR{$l = 1, \ldots, L$}
        \STATE \label{step start} 
        Based on the observed queue statuses, generate $\bm{X}_l$ according to $[1, G_1((l - 1)\times\Bar{K}), Q_1((l - 1)\times\Bar{K}), Q_1^2((l - 1)\times\Bar{K}), G_1((l - 1)\times\Bar{K})Q_1((l - 1)\times\Bar{K}), \ldots, G_N((l - 1)\times\Bar{K})Q_N((l - 1)\times\Bar{K})]$.
        \STATE \label{step policy generation} Draw $\mathcal{F}_l$\footnotemark  ~from the policy defined as
        \begin{small}
        \begin{align}
            \pi_l(f|\bm{X}_l) = (1 - \gamma) \frac{\exp\left( \eta \displaystyle\sum_{l' = 0}^{l - 1}\langle \bm{X}_l, \hat{\bm{\theta}}_{l', f} \rangle \right)}{\displaystyle\sum_{f' \neq f} \exp\left( \eta \displaystyle\sum_{l' = 0}^{l - 1}\langle \bm{X}_l, \hat{\bm{\theta}}_{l', f'} \rangle \right)} + \frac{\gamma}{|\mathcal{F}|}, 
        \end{align}
        \end{small}
        where $\langle \cdot, \cdot \rangle$ stands for the inner product in the Euclidean space as specified in \cite{anton_elementary_1987}, $\hat{\theta}_{l, f}$ is the estimated linear loss vector for choosing $f$ as the slicing decision $\mathcal{F}_l$; 
        \FOR{$k \in [(l - 1)\times\Bar{K} + 1, l \times \Bar{K}]$}
            \STATE Run Algorithm \ref{alg_psum_bcd} and acquire $F^*_k(\mathcal{F}_l)$;
        \ENDFOR
        \STATE \label{step end} Update $\hat{\bm{\theta}}_{l, f}$ according to
        \begin{small}
        \begin{align}
            \hat{\bm{\theta}}_{l, f} = \frac{\mathbb{I}_{\{ \mathcal{F}_l = f \}}}{\pi_l(f|\bm{X}_l)}   \frac{\frac{1}{\Bar{K}}\displaystyle\sum_{k = (l - 1) \times \Bar{K} + 1}^{l \times \Bar{K}}  F^*_k(\mathcal{F}_l)}{{\bm{X}_l}^T {\bm{X}_l}} \bm{X}_l, 
        \end{align}    
        \end{small}
        where $\mathbb{I}_{\{\cdot\}}$ stands for the indicator function. 
    \ENDFOR
    \end{algorithmic}
\end{algorithm}
\footnotetext{Noting that flat fading is widely adopted in sufficient number of orthogonal slicing works such as \cite{kwak_dynamic_2017, tang_service_2019, yang_how_2021}, therefore we focus on deriving the adequate quantity of sub-channels (i.e., $|\mathcal{F}_l|$) allocated to each slices. }
\begin{figure}[!t]
\includegraphics[width=0.7\linewidth, trim=10 20 10 0, clip]{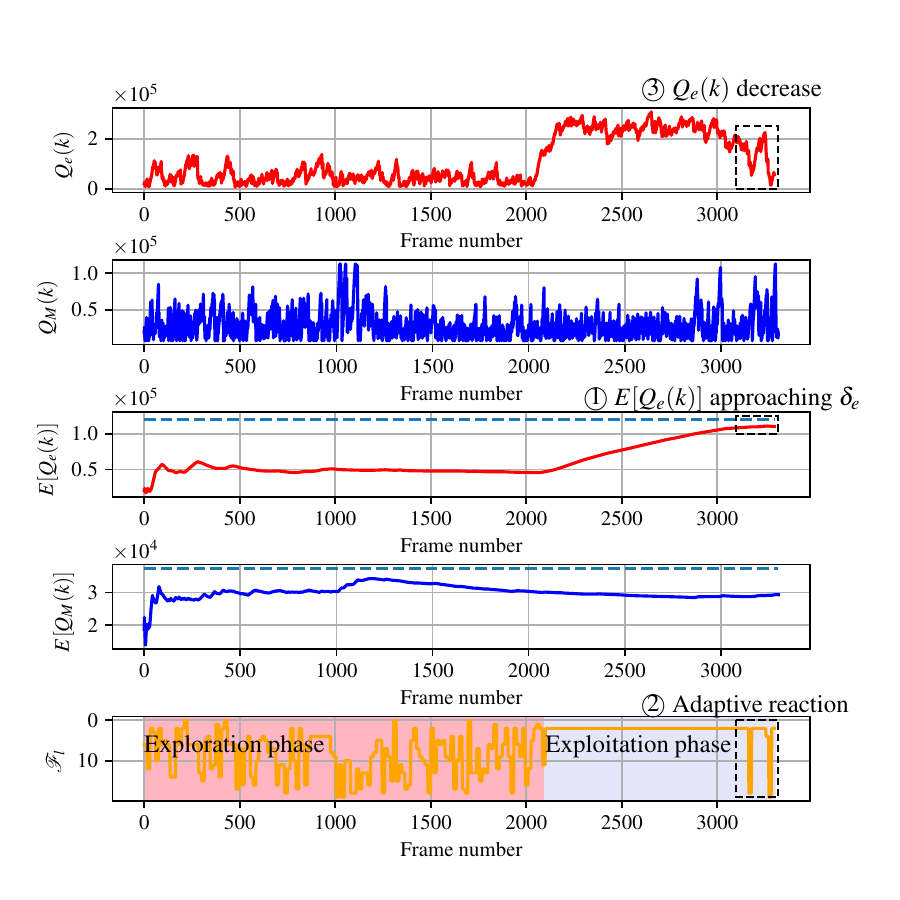}\
\centering
\caption{Slice configuration $\mathcal{F}_l$ and other system dynamics (e.g., $Q_e(k)$, $Q_M(k)$ denoting average backlog for eMBB and MBBLL users, respectively, and their time-average $\mathbb{E}[Q_e(k)]$, $\mathbb{E}[Q_M(k)]$). After entering exploitation stage, as $\textcircled{1}$ the $\mathbb{E}[Q_e(k)]$ approaching $\delta_e$, $\textcircled{2}$ our proposed Ad2S algorithm adaptively slices in favor of honoring legacy QoS, $\textcircled{3}$ $Q_e(k)$ decreases in response. }
\label{fig:queue_vs_slice}
\vspace{-3mm}
\end{figure}
We intuitively present the mechanisms of this algorithm through Fig.~\ref{fig:queue_vs_slice}. As the expected queueing backlog for legacy applications approaches the threshold, our algorithm senses the cumulation of both virtual and queueing backlog in terms of context as specified in step \ref{step start} and adaptively adjusts the slicing configuration in reaction, which then leads to efficient backlog reduction and QoS restoration for legacy users.

\section{Dual-timescale Resource Allocation Under Non-stationary Case}\label{non-sta}

In this section, we reconsider the drift-based Problem~\ref{prob Lyapunov} under non-stationary conditions. Although we can apply Algorithm~\ref{alg_psum_bcd} to obtain $F^*_k(\mathcal{F}_l)$, the non-stationary condition \eqref{pathloss_ar} renders Algorithm~\ref{stationary_alg} ineffective due to its slow convergence and sub-optimal issues as pointed out by \cite{seldin_evaluation_2013, allesiardo_non-stationary_2017}. 
Specifically, the estimation of $\frac{1}{\Bar{K}}\sum_{k = (l - 1) \times \Bar{K} + 1}^{l \times \Bar{K}} F^*_k(\mathcal{F}_l)$ is no longer accurate with the non-stationary channel introduced, leading to these issues. 
To address this and make it mathematically tractable, we refine previously proposed Ad2S with a concise contextual enhancement. In particular, we propose a Kalman filter-based dynamic tracking mechanism to obtain an estimated pathloss component $\hat{\bm\mu}(l)$ and discuss the network slicing refinement in what follows.

\subsection{Kalman Filter-based Dynamic Tracking}\label{section ME-KF}
Different from the previous case considering stationary channel statistics as given by \eqref{stationary assumption}, we propose a measurement extrapolation-Kalman filter (ME-KF) based solution to predict the pathloss component of the $l^{\mathrm{th}}$ super-frame, $\hat{\bm\mu}(l)$. Detailed ME-KF-based algorithms are listed below.

\subsubsection{Initialization Stage}
Initialize $\hat{\bm{A}}(0) = \bm{A}^-(l) = \bm{I}_{N}$, $\hat{\bm{\mu}}(0) = \tilde{\bm{\mu}}(0) = \bm{\mu}^-(0) = \tilde{\bm{\mu}}(0)^{\mathrm{ext}} = \bm{1}_{N}$ and $\bm{P}^-(0) = \bm{P}^+(0) = \bm{I}_{2N}$, where $\mathrm{ext}$ stands for ``extrapolated''.

\subsubsection{Prediction Stage}
Update $\bm{A}^-(l)$ and $\bm{\mu}^-(l)$ according to \eqref{pathloss_ar} and predict the covariance matrix $\bm{P}^-(l)$ as follows: 

\begin{small}
\vspace{-3mm}
\begin{eqnarray}
    \bm{A}^-(l)  & = & \beta \cdot \hat{\bm{A}}(l-1), \\
    \bm{\mu}^-(l) & = & \hat{\bm{A}}(l - 1) \cdot \hat{\bm{\mu}}(l-1), \\
    \bm{P}^-(l) & = & \bm{H}(l - 1) \cdot \bm{P}^+(l - 1) \cdot {\bm{H}(l - 1)}^{T} \nonumber\\
    && +
    \begin{bmatrix}
    \bm{Q}_A & ~\\
    ~ & \bm{Q}_{\mu}
    \end{bmatrix}. \label{prediction_stage_covariance}
\end{eqnarray}    
\end{small}

\noindent where $
    \bm{H}(l - 1) =  \begin{bmatrix}
    \beta \bm{I}_{N} & \bm{0}_{N}\\
    \bm{0}_{N} & \hat{\bm{A}}(l - 1)
    \end{bmatrix}.
$

\subsubsection{Extrapolation Stage}
Noted that $\tilde{\bm\mu}(l)$ is not observable at frame $(l - 1) \times \Bar{K} + 1$, rendering problem challenging with the measurement delay issue. 
Based on the optimal extrapolation in the Larsen's method \cite{larsen_incorporation_1998}, we extrapolate the delayed a posteriori measurement with $\tilde{\bm{\mu}}(l)^{ext} = \tilde{\bm{\mu}}(l - 1) + \bm{\mu}^-(l) - \bm{\mu}^-(l - 1)$.

\subsubsection{Fusion Stage}
In the fusion stage, we denote $\bm{K}(l) = \bm{M} \bm{P}_{[1:2N; N+1:2N]}^-(l - 1) {\bm{P}_{[N+1:2N; N+1:2N]}^-(l - 1)}^{-1}, \bm{M} = (\bm{I} - \begin{bmatrix}
        \bm{0} & \bm{K}(l)
    \end{bmatrix}) \bm{H}(l - 1)$ to be the Kalman gain of the estimator and extrapolation coefficient, respectively. 
This stage fusions the a priori $\bm{A}^-(l)$ and $\bm{\mu}^-(l)$ with extrapolated measurement $\tilde{\bm{\mu}}(l)^{ext}$ and updates as 

\begin{small}
\vspace{-3mm}
    \begin{eqnarray}
        \hat{\bm{A}}(l)&\!=\!&\bm{A}^-(l) + \textrm{diag}\left(\bm{K}_{[1:N; 1:N]}(l) \cdot \left( \tilde{\bm{\mu}}(l)^{ext} - \hat{\bm{\mu}}(l - 1) \right)\right), 
\nonumber\\
        \hat{\bm{\mu}}(l)&\!=\!&\bm{\mu}^-(l) + \bm{K}_{[N+1:2N; 1:N]}(l) \cdot \left( \tilde{\bm{\mu}}(l)^{ext} - \hat{\bm{\mu}}(l - 1) \right), 
\nonumber\\
        \bm{P}^+(l)&\!=\!&\bm{P}^-(l) - \bm{K}(l) \cdot \bm{P}_{[N+1:2N; 1:2N]}^-(l - 1) \cdot \bm{M}^T. 
\nonumber
    \end{eqnarray}
\end{small}

\noindent Output $\hat{\bm{A}}(l)$ and $\hat{\bm{\mu}}(l)$ are therefore updated with delayed input $\tilde{\bm\mu}(l - 1)$. 
The above ME-KF procedures track channel with fused a priori information and extrapolated a posteriori information. 
Therefore, accurate channel tracking is performed in a predictive manner. 
The convergence properties of this ME-KF based dynamic tracking are given as follows. 

\begin{Lem}[Convergence of ME-KF]\label{ME-KF_convergence}
    Under the non-stationary variation evolving as $\bm{A}(l) = \beta \bm{A}(l - 1) + \bm{N}_{\bm{A}}$ as listed in assumptions, with $\textrm{tr}(\bm{Q}_{\bm{A}}), \textrm{tr}(\bm{Q}_{\bm{A}}) < \infty$, when $l \to \infty$, the asymptotical error covariance matrix $\bm{P}^*$ is denoted as 

\begin{small}
\vspace{-3mm}
\begin{align}
    & \bm{P}^* \triangleq \lim_{l \to \infty} \mathbb{E}[\bm{P}(l)] = \bm{F}_1^* \bm{P}^* {\bm{F}_1^*}^T + \begin{bmatrix}
        \bm{Q}_{\bm{A}} & \bm{0}\\
        \bm{0} & \bm{Q}_{\bm{\mu}}
    \end{bmatrix} - \nonumber\\
    & \bm{F}_1^* \bm{P}^* \begin{bmatrix}
            \bm{0}\\
            \bm{I}
        \end{bmatrix} {\left( \begin{bmatrix}
            \bm{0} & \bm{I}
        \end{bmatrix} \bm{P}^-(l - 1) \begin{bmatrix}
            \bm{0}\\
            \bm{I}
        \end{bmatrix}
        \right)}^{-1} \begin{bmatrix}
            \bm{0} & \bm{I}
        \end{bmatrix} \bm{P}^* {\bm{F}_1^*}^T, 
\label{riccati equation}
\end{align}
\end{small}

\noindent where $\bm{F}_1^* \triangleq \displaystyle\lim_{l \to \infty} \mathbb{E}[\bm{F_1}(l)] = \begin{bmatrix}
    \beta\bm{I}_{N \times N} & \bm{0}_{N \times N}\\
    \bm{0}_{N \times N} & \bm{0}_{N \times N}
\end{bmatrix}$. \eqref{riccati equation} is an discrete Riccatti equation as specified in \cite{tkailauth_linear_2000}. Since $\begin{bmatrix}
        \bm{Q}_{\bm{A}} & \bm{0}\\
        \bm{0} & \bm{Q}_{\bm{\mu}}
    \end{bmatrix}$ is positive definite, both $\bm{F}_1(l)$ and $[\bm{0} \ \bm{I}]$ are observable, and that $\displaystyle\lim_{l \to \infty} \rho(\bm{F}_1(l)) < 1$, where $\rho(\bm{F}_1(l)) = \max\{|\lambda|\}$ stands for the spectral radius, $\lambda$ is an eigenvalue of $\bm{F}_1(l )$, $\bm{P}^*$ satisfies the properties of having a unique positive definite solution. Then, proposed ME-KF schemes converge to a steady-state error covariance $\bm{P}^*$ as $l \to \infty$. 
\end{Lem}
Then, with the above ME-KF procedures, we can dynamically track $\bm{\mu}(l)$ with converged covariance and obtain an unbiased estimation of it as proved in \cite{chagas_extrapolation_2016}. As further illustrated in Fig.~\ref{fig:ekf}, beyond straightforwardly tracking $\bm{\mu}(l)$ with a priori autoregressive models \cite{vinogradova_estimating_2022}, our proposed ME-KF scheme introduces extrapolated a posteriori observation and achieves lower estimation error against $\bm{N}_{\bm{A}}$ and $\bm{N}_{\bm{\mu}}$.

\begin{figure}[!t]
\centering
\centering
\includegraphics[width=0.75\linewidth, trim=40 20 45 0, clip]{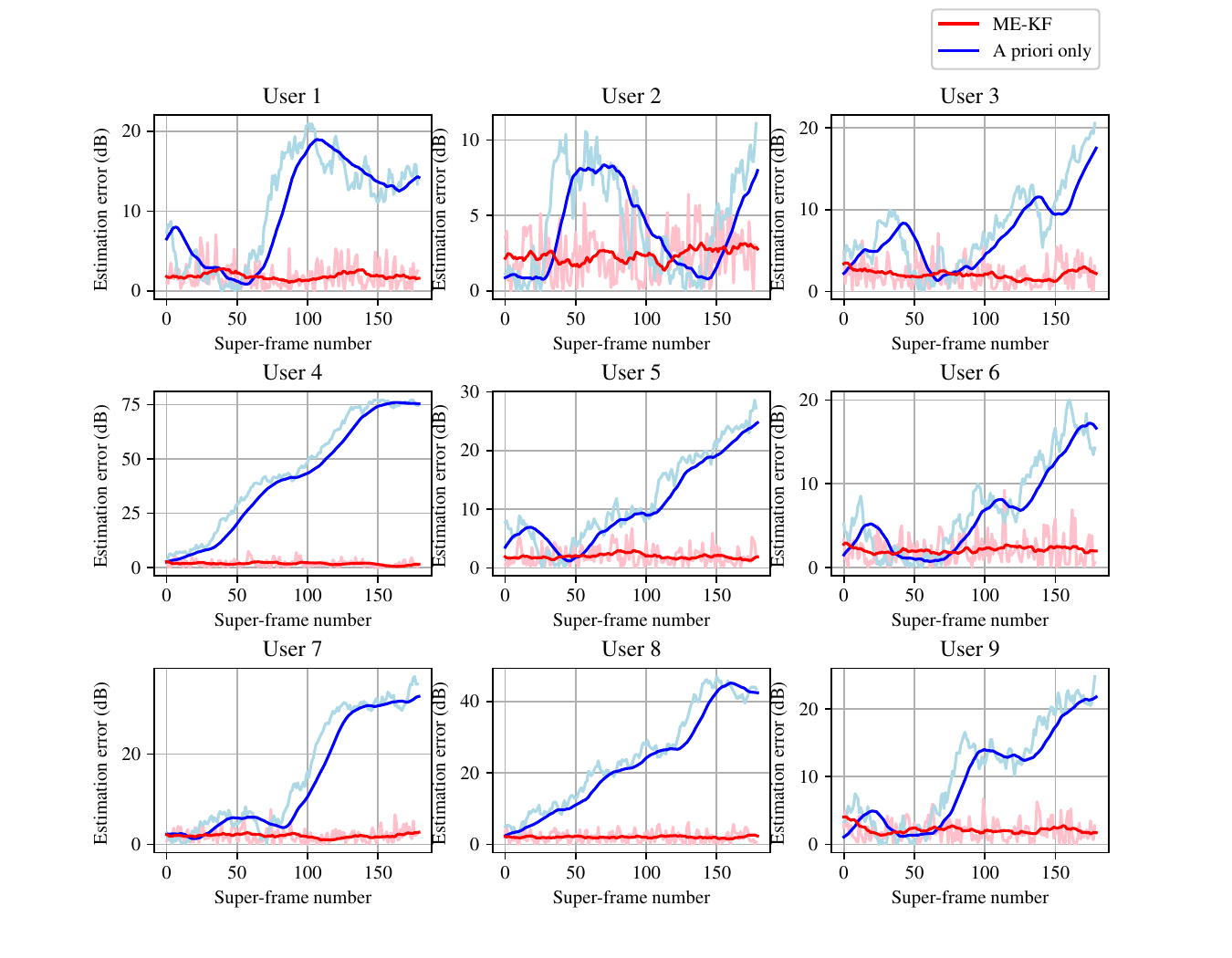}
\caption{{\color{black}Estimation error, e.g. $|\hat{\mu_n}(l) - \mu_n(l)|$ of a priori based \cite{vinogradova_estimating_2022} and proposed ME-KF scheme for each user.}}
\vspace{-5mm}
\label{fig:ekf}
\end{figure}

\subsection{Network Slicing Refinement with Predictive Results}
With the obtained $\hat{\bm{\mu}}(l)$, we focus on tailoring Algorithm~\ref{stationary_alg} into adaptation to non-stationary channel statistics. To be specific, to eliminate the non-stationary throughput variation over super-frames, we adopt the following refinements to Alg.~\ref{stationary_alg}, and name it after adaptive adversarial slicing-non-stationary refinement (Ad2S-NR) algorithm. 
\begin{itemize}
    \item First, at the start of each super-frame, i.e., ahead of step \ref{step start}, run ME-KF and obtain $\hat{\bm{\mu}}(l)$. 
    
    \item Second, we redefine the contextual vector $\bm{X}(l) = [1, G_1((l - 1)\times\Bar{K}), Q_1((l - 1)\times\Bar{K}), Q_1^2((l - 1)\times\Bar{K}), G_1((l - 1)\times\Bar{K})Q_1((l - 1)\times\Bar{K}), \hat{R}_1(l), \hat{R}_1^2(l),  Q_1((l - 1)\times\Bar{K}) \hat{R}_1(l),  G_1((l - 1)\times\Bar{K}) \hat{R}_1(l), \ldots, G_N((l - 1)\times\Bar{K}) \hat{R}_N(l)]$, where 
    \begin{small}
    \begin{align}
        \hat{R}_n(l) = \left\{\begin{array}{ll}
            2 \hat{\mu}_n(l) \cdot \log_2(e), & \hat{\mu}_n(l) \geq \frac{\ln\left( \frac{\tau |\mathcal{F}|}{P_{\mathrm{tot}}} \right)}{2}, \\
            e^{-\bm{P}_{N+n, N+n}(l)} \cdot e^{2\hat{\mu}_n(l)} & \hat{\mu}_n(l) < \frac{\ln\left( \frac{\tau |\mathcal{F}|}{P_{\mathrm{tot}}} \right)}{2}, 
        \end{array}\right.
    \label{throughput_approximate}
    \end{align} 
    \end{small}
    is the approximated unit spectral efficiency level related with $\hat{\mu}_n(l)$, $\tau$ denotes the threshold.
    \item Third, we resize the per-subcarrier loss as $\hat{\bm{\theta}}_{l, f} \in \mathbb{R}^{1 + 8N}, \forall f \in [1, \ldots, |\mathcal{F}|]$. 
\end{itemize}

\section{Performance Analysis}\label{performance analysis}
{\color{black}
    In this section, we analyze the asymptotical optimality and computational complexity properties with both Ad2S and Ad2S-NR. 
}
\begin{Thm}[Asymptotical Optimality for Ad2S and Ad2S-NR]
    By setting learning rate and exploration parameter as $\eta = L^{-\frac{2}{3}}(|\mathcal{F}|(1 + 4N))^{-\frac{1}{3}}(\log (|\mathcal{F}|))^{\frac{2}{3}}$ and $\gamma = L^{-\frac{1}{3}}(|\mathcal{F}|(1 + 4N)\log (|\mathcal{F}|))^{\frac{1}{3}}$, respectively, for Ad2S, the asymptotic optimality in terms of cumulative regret\footnote{The cumulative regret represents the gap between the historical slicing choices and the optimal slicing choice. As \( l \to \infty \), if the regret grows sublinearly, the system tends toward asymptotic optimality \cite{prabhu_sequential_2022}. } is bounded as follows, 
    \begin{align}
        \mathrm{Regret}{L}^{\textrm{Ad2S}} \leq 5L^{\frac{2}{3}}(|\mathcal{F}|(1 + 4N)\log|\mathcal{F}|)^{\frac{1}{3}} F_{\mathrm{max}}. \label{sta_regret}
    \end{align}
    where $F_{\mathrm{max}}$ is a constant that satisfies $F_{\mathrm{max}} \geq \frac{1}{\Bar{K}}\sum_{k = (l - 1) \times \Bar{K} + 1}^{l \times \Bar{K}}F^*_k(\mathcal{F}_l), \forall l$. 
    
    While under non-stationary cases which $\bm{\mu}(l)$ evolves as \eqref{pathloss_ar}, by setting learning rate and exploration parameter as $\eta = L^{-\frac{2}{3}}(|\mathcal{F}|(1 + 8N))^{-\frac{1}{3}}(\log |\mathcal{F}|)^{\frac{2}{3}}$ and $\gamma = L^{-\frac{1}{3}}(|\mathcal{F}|(1 + 8N)\log |\mathcal{F}|)^{\frac{1}{3}}$, respectively, for Ad2S-NR, the asymptotic optimality in terms of cumulative regret is bounded as follows, 
    \begin{align}\label{non_sta_regret}
        \mathrm{Regret}_{L}^{\textrm{Ad2S-NR}} \leq & 5L^{\frac{2}{3}}(|\mathcal{F}|(1 + 8N)\log|\mathcal{F}|)^{\frac{1}{3}} F_{\mathrm{max}}. 
    \end{align}
    where both $\mathrm{Regret}_{L}^{\textrm{Ad2S}}$ and $\mathrm{Regret}{L}^{\textrm{Ad2S-NR}}$ are denoted as 
    \begin{align}\label{regret}
        \mathrm{Regret}_{L} \triangleq & \max_{\pi \in \Pi}\mathbb{E}\Bigg[ \sum_{l = 1}^{L} (\frac{1}{\Bar{K}}\sum_{k = (l - 1) \times \Bar{K} + 1}^{l \times \Bar{K}}F^*_k(\mathcal{F}_l) \nonumber\\
        & - \frac{1}{\Bar{K}}\sum_{k = (l - 1) \times \Bar{K} + 1}^{l \times \Bar{K}}F^*_k(\bm{\pi}(\bm{X}_l)) \Bigg], 
    \end{align} 
    $\bm{\pi}(\bm{X}_l) = [\pi(1 \vert \bm{X}_l), \ldots, \pi(|\mathcal{F}| \vert \bm{X}_l)]$ is the strategy vector given $\bm{X}_l$ \cite{neu_efficient_2020}. 
    
\label{theorem_regret}
\end{Thm}

\begin{proof}
    Please refer to Appendix~\ref{appendix:theorem_regret} for the proof. 
\end{proof}

Theorem \ref{theorem_regret} shows that Ad2S and Ad2S-NR achieve sub-linear regret (e.g. $o(L)$) performance under both stationary and non-stationary scenarios. Furthermore, approximately $1.3  \times$ additional regret is introduced due to the extended contextual dimensions under non-stationary cases, Fig.~\ref{fig:regret} can be referred as the numerical proof. 
{\color{black}
\begin{Lem}[Computational complexity for Ad2S and Ad2S-NR]
\label{lem_comp_complexity}
    The overall computational complexity for running both Ad2S and Ad2S-NR at the start of super-frame $l$ follows $\mathcal{O}(|\mathcal{F}| (1 + 4N)^2)$ and $\mathcal{O}\left(|\mathcal{F}| (1 + 8N)^2 + N^3\right)$, respectively. 
\end{Lem}
\begin{proof}
    First, we analyze Ad2S's computational complexity. The policy generation step \ref{step policy generation} in Ad2S has the complexity of $\mathcal{O}(|\mathcal{F}| (1 + 4N))$. The parameter update step \ref{step end} in Ad2S has the complexity of $\mathcal{O}(|\mathcal{F}| (1 + 4N)^2)$. The overall computational complexity of Ad2S is dominated by step \ref{step end} and thus follows $\mathcal{O}(|\mathcal{F}| (1 + 4N)^2)$. 
    
    Then, we analyze Ad2S-NR's computational complexity. The computational complexity is divided into two parts for ME-KF and refined Ad2S, respectively. 
    \begin{enumerate}
        \item{\bf Complexity for ME-KF. } The prediction stage has the complexity of $\mathcal{O}(N^3)$ dominated by \eqref{prediction_stage_covariance}. The extrapolation and fusion stages also have the complexity of $\mathcal{O}(N^3)$ dominated by the update functions in the fusion stage. Then the overall computational complexity for ME-KF follows $\mathcal{O}(N^3)$. 
        \item{\bf Complexity for refined Ad2S. } Similar to Ad2S's computational complexity, refined Ad2S has the overall complexity of $\mathcal{O}(|\mathcal{F}| (1 + 8N)^2)$. 
    \end{enumerate}
    Then, the overall computational complexity of Ad2S-NR follows $\mathcal{O}\left(|\mathcal{F}| (1 + 8N)^2 + N^3\right)$. 
\end{proof}
Lemma~\ref{lem_comp_complexity} shows that both Ad2S and Ad2S-NR mainly encounter complexity bottlenecks of sub-channel number and user number. However, with coarse-grained slicing chunks \cite{zanzi_laco_2021} ($|\mathcal{F}| \ll N$), Ad2S-NR has the complexity of $\mathcal{O}(N^3)$ and its computational speed is primarily limited by the number of users. 
}

\section{Simulation Results}\label{simulation results}
In this section, we provide some numerical examples to demonstrate the advantages of our proposed schemes. Specifically, we evaluate the system performance in a multi-user downlink scenario, where the standard deviation of log-normal shadowing is set to be 5 dB \cite{noauthor_study_2024}. 
{\color{black}
In addition, to address dynamic and bursty traffic arrival patterns, we adopt the traffic model from \cite{yang_how_2021}, parameterized with inter-batch arrival $\lambda_a$, which is subsequently validated in our experimental evaluations.}
The rest detailed parameters are listed in Table~\ref{tab:simsetting} \cite{noauthor_system_2024}. 
{\color{black}
\begin{table}[t!]
\vspace{-3mm}
\caption{List of Simulation Parameters.}
\label{tab:simsetting}
\begin{center}
\begin{tabular}{c p{1.8cm}|c p{2.5cm}}
\toprule
\textbf{Parameters}	& \textbf{Value} & \textbf{Parameters}	& \textbf{Value}\\
\midrule
$N_{e}$ & $\{3, 4, \ldots, 7\}$ & $\lambda_e$ & 10000 packets/frame \\
\hline
$N_{U}$ & $\{2, 3, \ldots, 6\}$ & $\lambda_U$ & 38 packets/frame \\
\hline
$N_{M}$ & $\{1, 2, \ldots, 5\}$ & $\lambda_M$ & 12500 packets/frame \\
\hline
$d_e$ & 120 ms & $d_M$ & 30 ms \\
\hline
B & 360 kHz & $\eta$ & $1.25e{-4}$ packet$\cdot$s/bit\\
\hline
$\tau$ & 1 dB & $\Bar{K}$ & 100\\
\hline
$|\mathcal{F}|$ & $\{20, 24, \ldots, 36\}$  & $P_{tot}$  & $\{39, 40, \ldots, 43\}$ dBm\\
\hline
$\bm Q_{\bm{A}}$ & $1 \times 10^{-4}\bm{I}_N$ & $\bm Q_{\mu}$ & $\bm{I}_N$\\
\hline
$\omega_Q$ & $5\times10^{-8}$ & $\omega_T$ & $1\times10^{-3}$\\
\bottomrule
\end{tabular}
\end{center}
\vspace{-3mm}
\end{table} 
}

Some recently proposed schemes are referred as baselines, which are listed as follows. 
\begin{itemize}
    \item {{\bf Baseline 1: Non-adaptive slicing (NAdS) with dynamic resource allocation scheme (DRAS)\footnote{The puncturing mechanism is not considered in our system, therefore called ``DRAS''. } \cite{zhao_resource_2022}}, this scheme addresses Problem~\ref{prob frame} following a heuristic QoS oriented scheme, NAdS indicates that fixed slicing configuration with equal sub-channels provisioned for legacy and 6G slices, respectively. }
    \item {{\bf Baseline 2: NAdS with proposed RA scheme}, where NAdS for NS and PBRA for RA are performed. Then the following algorithms focus on benchmarking NS performance, and default using PBRA as RA policies. }
    \item {{\bf Baseline 3: Exponential-weight algorithm for exploration and exploitation (EXP3) \cite{auer_nonstochastic_2002} based network slicing}, this scheme employ classic EXP3 method to exploit the optimal arms (NS decisions) under adversarial scenarios. }
    \item {{\bf Baseline 4: Contextual upper confidence bound (UCB) \cite{li_contextual-bandit_2010, zanzi_laco_2021} based network slicing}, where the UCB is introduced to determine whether arms are fully explored, this scheme jointly considered the relationship between historic reward and context, and achieves promising sub-linear regret under sub-Gaussian distributed cases.}
    \item {{\bf Proposed:} Both Ad2S and Ad2S-NR for NS are performed under stationary and non-stationary scenarios, respectively. }
\end{itemize}

\begin{figure}[!t]
\vspace{-3mm}
\centering
\includegraphics[width=0.68\linewidth, trim=0 0 20 0, clip]{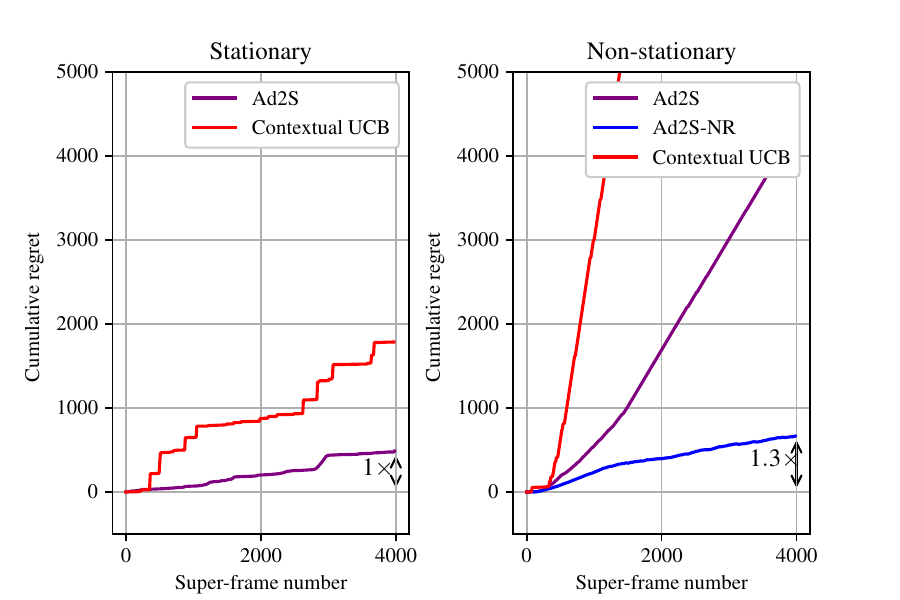}
\caption{Empirical cumulative regret. }
\vspace{-3mm}
\label{fig:regret}
\end{figure}

{\color{black}
\subsection{Non-stochastic and Non-stationary Robustness}
We first evaluate the performance of the proposed Ad2S and Ad2S-NR schemes in comparison with the contextual UCB method, using the empirical cumulative regret metric. 
As shown in Fig.~\ref{fig:regret}, the cumulative regret achieved by the proposed schemes is significantly improved. In particular, while both Ad2S and contextual UCB experience superlinear cumulative regret in non-stationary scenarios, the cumulative regret of Ad2S-NR remains sublinear. Moreover, the non-stationary regret overhead of Ad2S-NR is approximately $1.3 \times$, consistent with the bound established in Theorem~\ref{theorem_regret}. These findings highlight that although contextual UCB demonstrates strong performance in cases with sub-Gaussian reward, the non-stochastic effects arising from multi-timescale interactions present substantial challenges. On the other hand, even Ad2S with non-stochastic robustness, struggle to learn effectively under the challenges posed by non-stationary parameter forgetting. 
The incorporation of non-stochastic robustness and non-stationary refinement mechanisms in our proposed Ad2S-NR scheme effectively addresses these challenges, thereby facilitating sublinear convergence of cumulative regret.
}


\begin{figure*}[!t]
\centering
\includegraphics[width=5.2in, trim=0 20 0 0, clip]{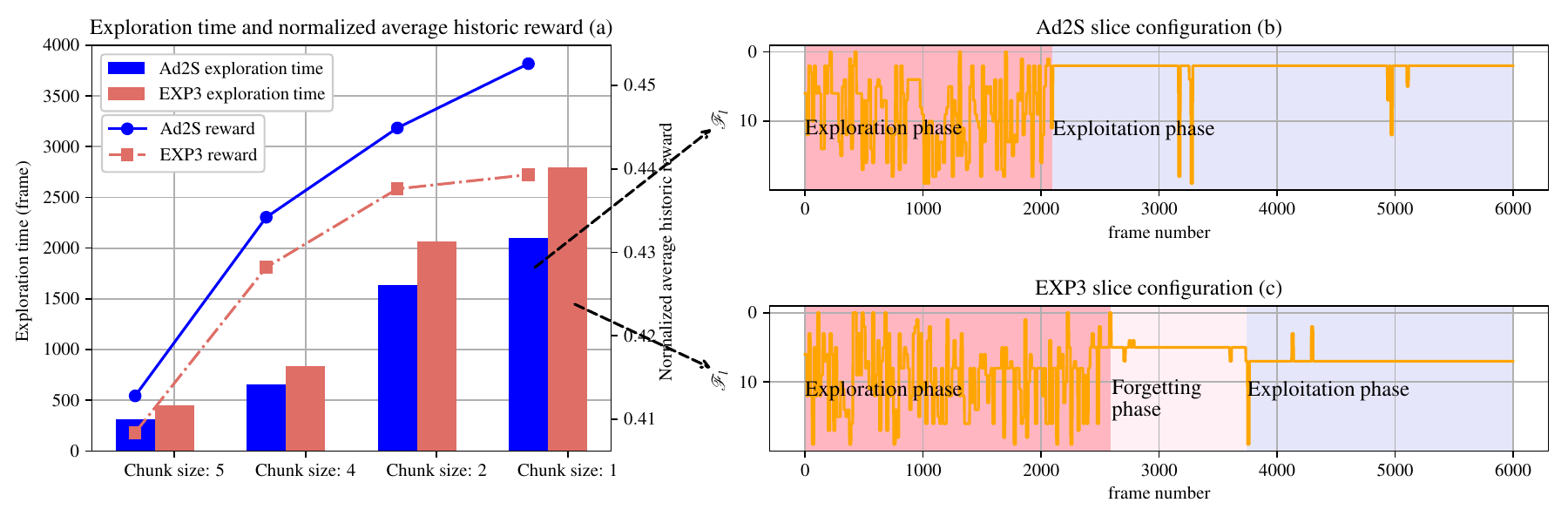}
\caption{{\color{black}(a) Impact of different chunk sizes on exploration time and average reward for Ad2S and EXP3 under identical parameters (e.g., $\eta = 1, \gamma = 0.5$). (b) Slice configuration behavior of Ad2S under chunk size of 1 sub-channel. (c) Slice configuration behavior of EXP3 under chunk size of 1 sub-channel. }}
\label{fig:exploration time}
\end{figure*}
\subsection{Adaptive QoS-aware Framework}
Second, we evaluate the convergence speed of Ad2S against EXP3. In this evaluation, we investigate the impact of chunk size on exploration durations, where the chunk size stands for the number of sub-channels grouped together when configuring slices. Evaluation is carried out with identical parameters $\eta, \gamma$ set for both Ad2S and EXP3. As shown in Fig.~\ref{fig:exploration time}(a), as the exploration time decreases monotonically with the chunk size, our proposed Ad2S experienced an average of $24.36\%$ less exploration time than EXP3. However, faster convergence comes at the cost of suboptimal reward. Consistently improved average reward for the Ad2S algorithm is demonstrated with numerical samples. 

By further investigating the slicing configuration behavior of our proposed Ad2S and EXP3 as illustrated in Fig.~\ref{fig:exploration time}(b) and Fig.~\ref{fig:exploration time}(c), we have that our proposed Ad2S scheme converges faster and respond to the temporal evolving backlog status of buffer and virtual queues in an agile manner. In contrast, EXP3 scheme takes $1.2 \times$ more number of iterations for the initial exploration, and additionally experiences the forgetting phase \cite{besbes_stochastic_2014} of more than 1000 frames before finally converging to an alter slicing configuration. These highlight that: rather than straightforwardly learning the reward, our proposed context-based scheme learns the parameter $\bm{\theta}$ in a more efficient QoS-aware manner. 

\begin{figure}[!t]
\centering
\includegraphics[width=0.65\linewidth, trim=20 0 20 0, clip]{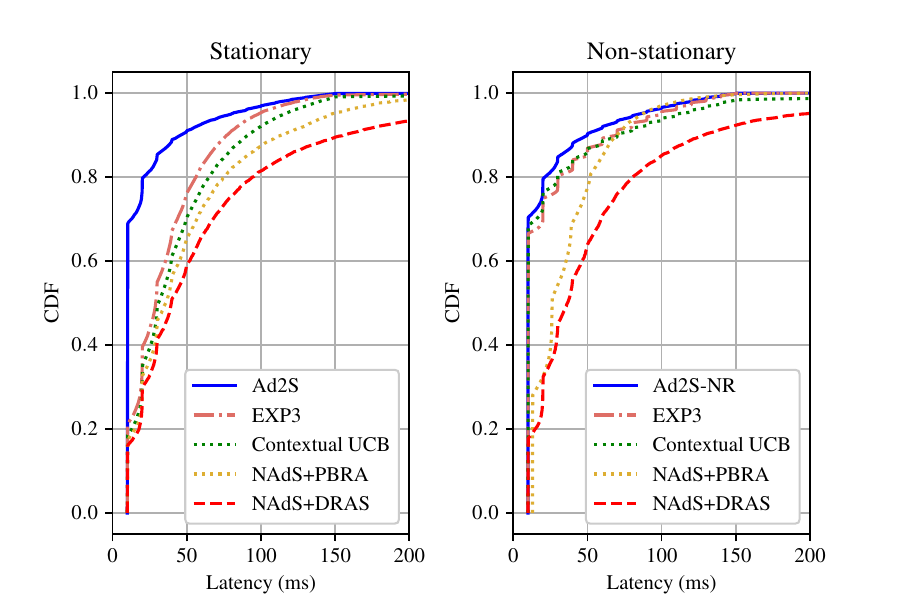}
\caption{Empirical Cumulative Distribution Function (CDF) of experienced latency. }
\label{fig:cdf}
\end{figure}
\subsection{Ad2S-NR Benchmarking}
Third, we evaluate the latency performance as shown in Fig.~\ref{fig:cdf}. Under stationary scenarios, our proposed schemes achieve improved latency with a mean delay equal to $\{20.72, 36.74, 43.59, 50.48, 70.44\}$(ms) for Ad2S, EXP3, Contextual UCB, NAdS+PBRA and NAdS+DRAS, respectively; while under non-stationary scenarios, a mean delay equal to $\{21.23, 26.04, 27.93, 27.59, 58.90\}$(ms) for Ad2S-NR, EXP3, Contextual UCB, NAdS+PBRA, and NAdS+DRAS, respectively. 
\begin{figure}[!t]
\centering
\includegraphics[width=0.7\linewidth, trim=20 0 0 0, clip]{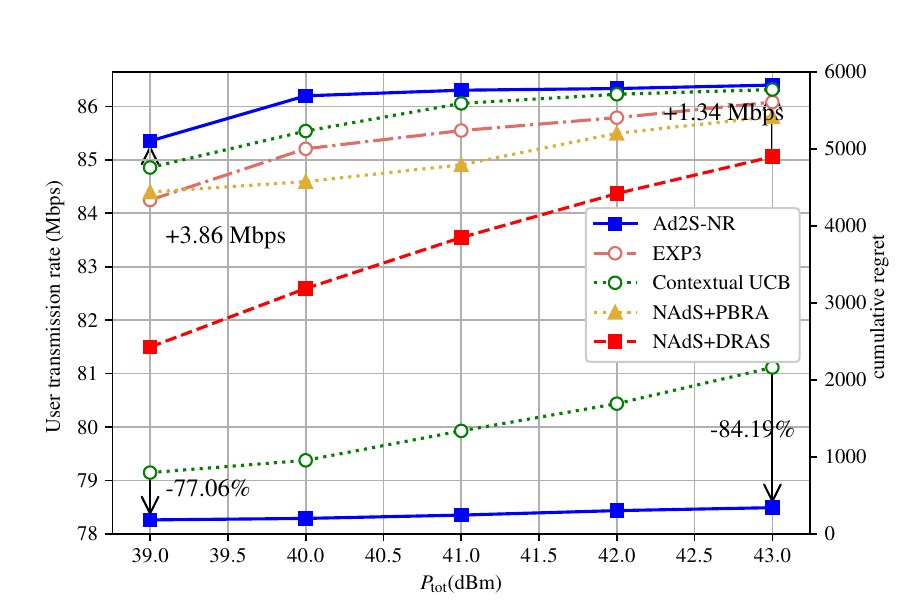}
\caption{{\color{black}Average transmission rate and cumulative regret versus total power budget. As shown in this figure, the transmission rate of our proposed schemes is superior to conventional schemes, as the total power budget increases. Cumulative regret of our proposed schemes is also consistently lower than the contextual bandits baseline. }}
\label{fig:rate_ptot}
\end{figure}

\begin{figure}[!t]
\centering
\includegraphics[width=0.7\linewidth, trim=20 0 20 0, clip]{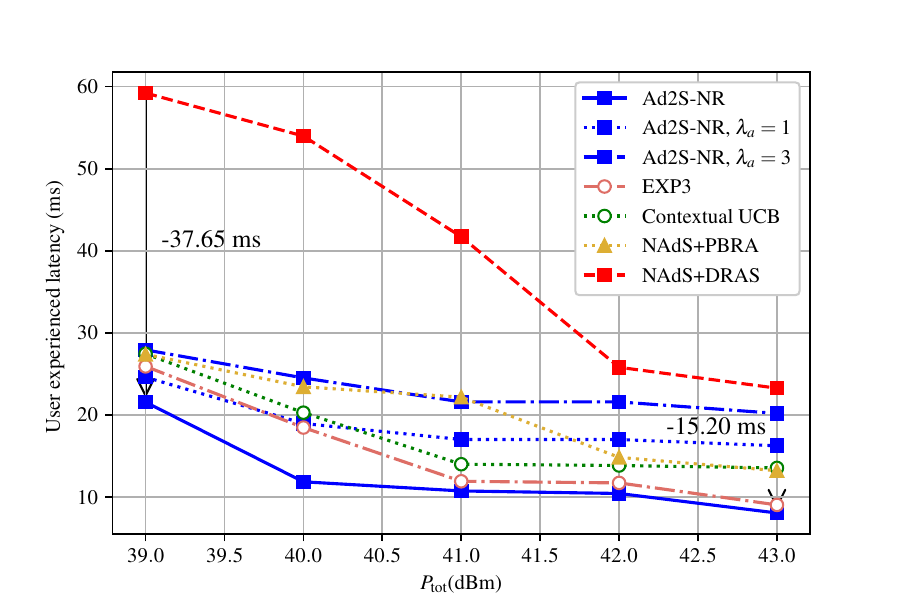}
\caption{{\color{black}Average latency versus total power budget. As shown in the figure, the average latency of our proposed scheme is smaller than conventional schemes, when the power budget increases. This figure also shows that increased level of traffic variation, in terms of $\lambda_a$ mentioned in \cite{yang_how_2021}, raises difficulties and introduces extra delay. }}
\label{fig:latency_ptot}
\end{figure}

\begin{figure}[!t]
\centering
\includegraphics[width=0.7\linewidth, trim=10 0 0 0, clip]{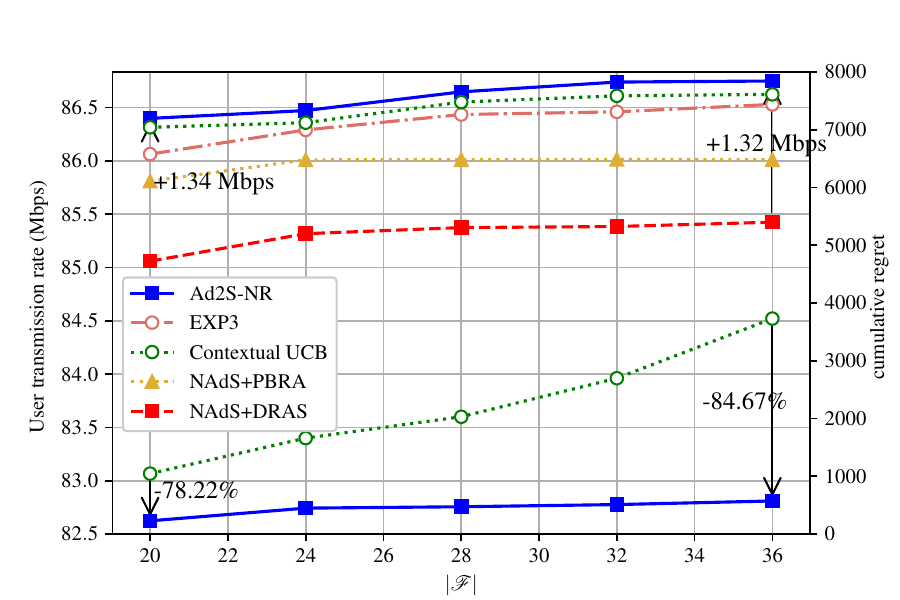}
\caption{{\color{black}Average transmission rate and cumulative regret versus number of sub-channels. As shown in this figure, our proposed scheme also shows regret and throughput improvement over conventional schemes, as the sub-channel number increases. As the action space scales up, our proposed schemes demonstrate consistent regret decrease. }}
\label{fig:rate_subchannel}
\end{figure}

\begin{figure}[!t]
\centering
\includegraphics[width=0.7\linewidth, trim=10 0 20 0, clip]{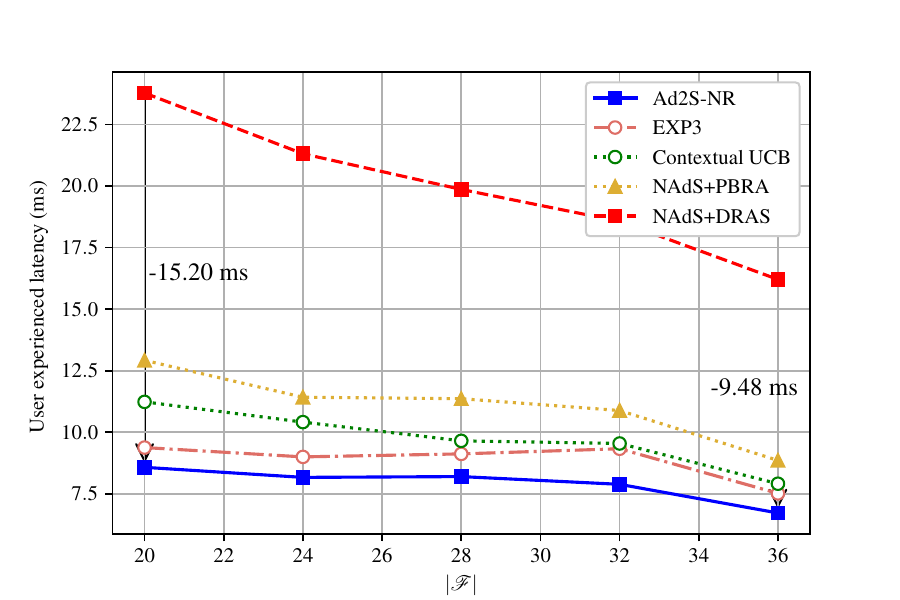}
\caption{{\color{black}Average latency versus number of sub-channels. As shown in this figure, the average latency of our proposed schemes is lower than baselines as well. Larger latency improvement is demonstrated under limited number of sub-channels, as our schemes dynamically utilize fragmented spectrum under scare sub-channel resource. }}
\label{fig:latency_subchannel}
\end{figure}

\begin{figure}[!t]
\centering
\includegraphics[width=0.7\linewidth, trim=20 0 0 0, clip]{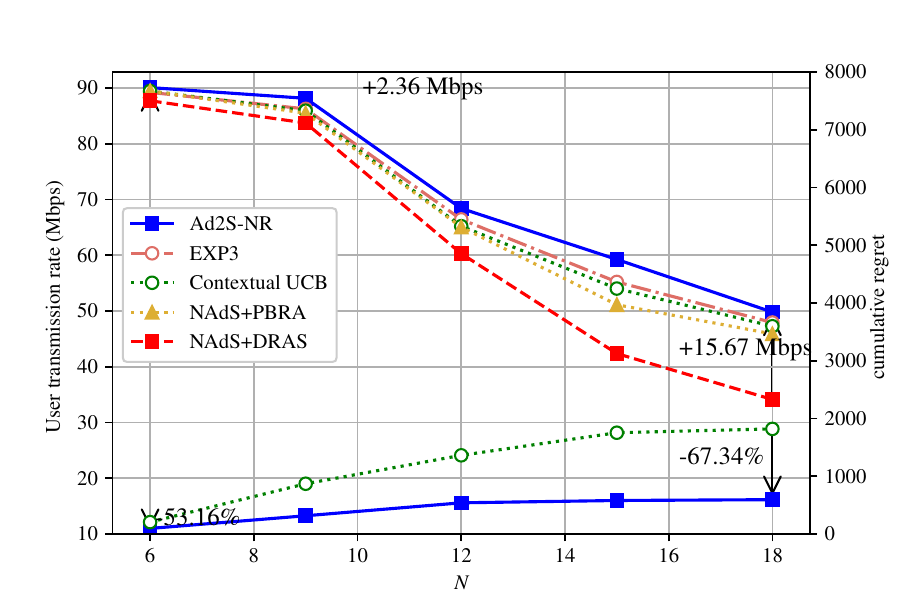}
\caption{{\color{black}Average transmission rate and cumulative regret versus number of users. As shown in this figure, the total transmission rate of our proposed scheme is still superior over conventional schemes, when the numbers of MBBLL, eMBB and URLLC users increase. And the regret is also smaller than the contextual bandits baseline. As the contextual vector scales up, decrease in regret is still robust with our proposed schemes. }}
\label{fig:rate_ue}
\end{figure}

\begin{figure}[!t]
\centering
\includegraphics[width=0.7\linewidth, trim=20 0 20 0, clip]{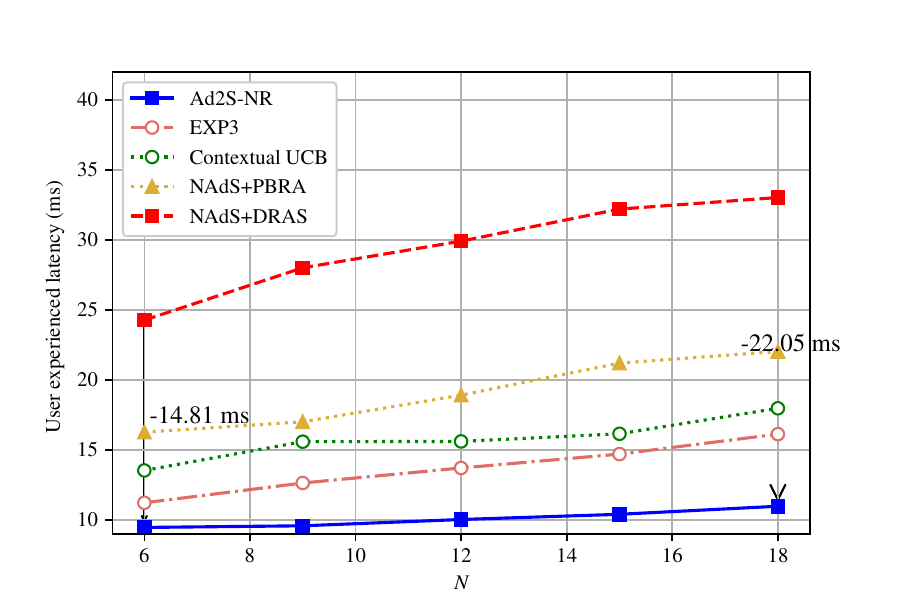}
\caption{{\color{black}Average latency versus number of users. As shown in this figure, the average user experienced latency of our proposed scheme is smaller as well, if compared with conventional schemes.}}
\label{fig:latency_ue}
\end{figure}

\begin{table}[h]
\vspace{8mm}
\setlength{\abovecaptionskip}{-0.2cm}   
\setlength{\belowcaptionskip}{-1cm}   
\caption{URLLC QoS satisfactory results.}
\begin{center}
\begin{tabular}{c c}
\toprule
\textbf{Baselines}	& \textbf{QoS satisfaction ratio}\\
\midrule
NAdS+DRAS & 80.3\% \\
NAdS+PBRA & 82.6\% \\
EXP3 & 89.0\% \\
Contextual UCB & 70.2\% \\
Ad2S-NR & 92.8\% \\
\bottomrule
\end{tabular}
\end{center}
\label{tab:urllc}
\end{table}
{\color{black}
Extended investigations on the impact of the power budget $P_{\mathrm{tot}}$, number of sub-channels $|\mathcal{F}|$ and number of users $N$ are carried out as shown in Fig.~\ref{fig:rate_ptot}-Fig.~\ref{fig:latency_ue}. 
We can observe that increasing the power and sub-channel budget leads to higher transmission rate and lower latency performance. 
Increasing number of users sacrifices transmission rate and latency performances and lead to increasing regrets. Strong traffic variation also leads to the delay increase. 
Notably, Ad2S-NR consistently achieves 53.16-84.67\% lower regret, 1.32-15.67 Mbps higher throughput, and 9.48-37.65 ms latency reduction compared to benchmarks across all test configurations. The traffic variation impact analysis in Fig.~\ref{fig:latency_ptot} further shows our scheme maintains $< 15\%$ performance degradation under $3 \times$ traffic fluctuation intensity, demonstrating environmental adaptability.
}
 Also, 92.8\% of URLLC QoS satisfaction ratio, e.g., $Pr(\eta r_n(k) \geq Q_n(k))$, is achieved with the implementation of Ad2S-NR, which outperforms other baselines as shown in Table~\ref{tab:urllc}. 

In conclusion, our proposed Ad2S-NR outperforms previous schemes with higher throughput, lower latency, and faster convergence speed. 

\subsection{Throughput Maximization versus Coupled QoS Guarantees}
After validating the numerical performance with benchmarks, we move on to investigate the trade-off between overall throughput and QoS guarantees introduced by the applied framework. 

First, we numerically verify the effectiveness of the proposed Lyapunov-based framework in stabilizing queue backlogs, thereby ensuring QoS satisfaction. Fig.~\ref{fig:E_queue_backlog} illustrates the average queue backlog for both eMBB and MBBLL users, computed as $\frac{1}{K}\sum_{k = 0}^{K - 1} Q_n(k)$ for $i \in \{e, M\}$. The results demonstrate that the backlog remains within the QoS thresholds defined by $\delta{e, M}$, which aligns with the theoretical guarantees established in Lemma~\ref{lem_dmu_long_term}. These results confirm that the Lyapunov-based optimization framework is capable of meeting QoS requirements for both legacy and next-generation applications.


{\color{black}
Second, with the 3D illustration of stability and rate trade-off as shown in Fig.~\ref{fig:latency vs total throughput}, we give guidelines for parameter tuning in practice as what follows. First, mark the typical telecom operator KPI values mentioned in the \cite{poulkov_future_2019} (60Mbps average user throughput and 5ms delay outage) on the graph and set them as the target. By plotting the contours of delay and throughput, we obtain the parameter selection range, bounded by the critical points of $(\omega_Q, \omega_T) = \{(5 \times 10^{-8}, 0.0009), (6 \times 10^{-8}, 0.001), (8 \times 10^{-8}, 0.001), (6 \times 10^{-8}, 0.0004)\}$, thus the ratio $\omega_Q/\omega_T$ is chosen within the range of $6 \times 10^{-5} - 8 \times 10^{-5}$ \cite{bracciale_lyapunov_2020}. Through fine-tuning within this range, telecom operators can further optimize for either latency or throughput demands while ensuring the aforementioned KPIs are met.}

{\color{black}
In addition, we further validated our algorithm's performance under bursty traffic patterns, which is still promising; also, we demonstrated the superiority of our schemes compared to DRL-based approaches. To maintain the paper's conciseness, these results have been included in Appendix~\ref{performance_bursty} and \ref{performance_ppo}. 
}
\begin{figure}[!t]
\centering
\includegraphics[width=0.74\linewidth, trim=20 10 20 20, clip]{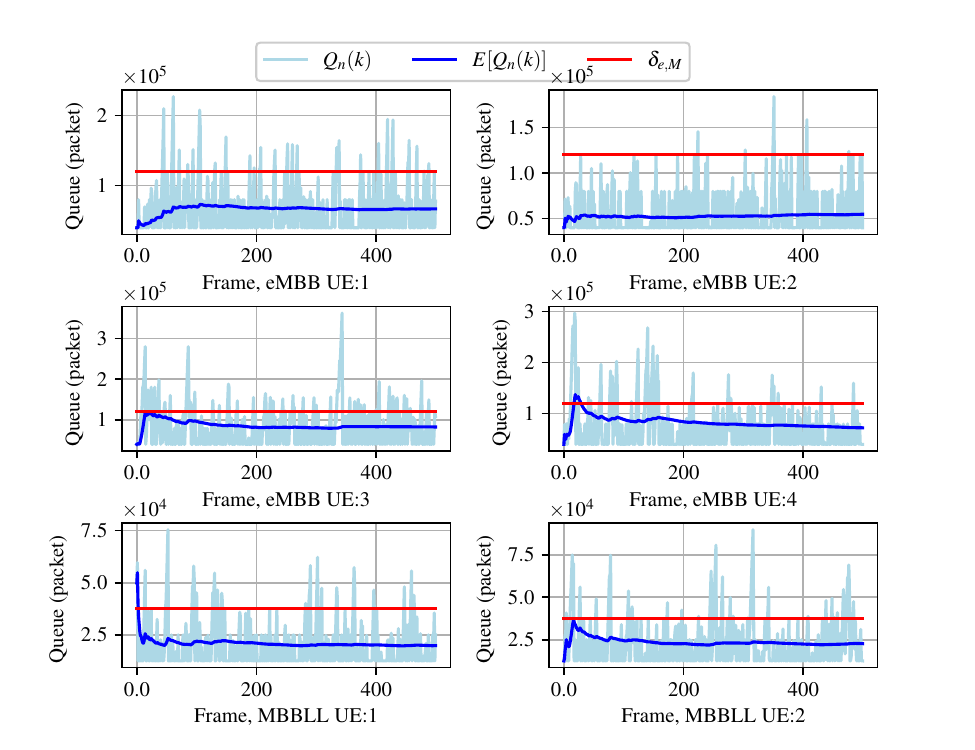}
\caption{{\color{black}Tendency of average queue backlog for eMBB and MBBLL UEs with proposed algorithm versus frame index. }}
\label{fig:E_queue_backlog}
\vspace{-3mm}
\end{figure}
\begin{figure}[!t]
\centering
\includegraphics[width=0.7\linewidth, trim=60 70 10 80, clip]{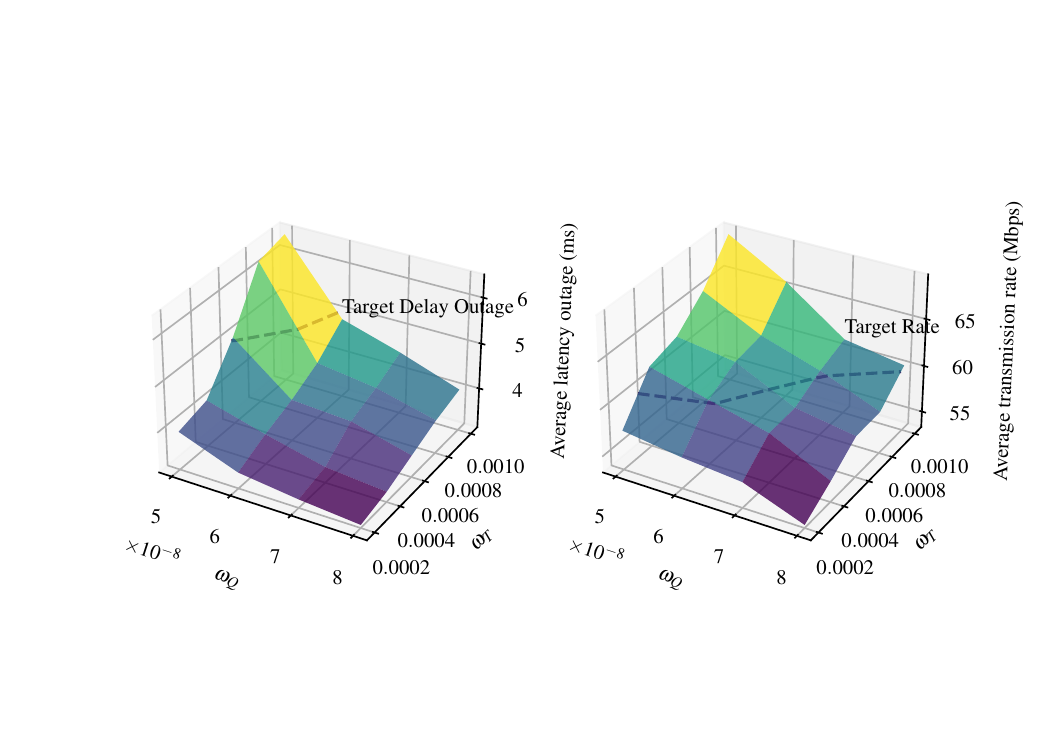}
\caption{{\color{black}A 3D illustration of average latency outage and average transmission rate versus varying $\omega_Q$ and $\omega_T$. As shown in this figure, the delay outage decrease as $\omega_Q$ increases. While as $\omega_T$ increases, higher average transmission rate are promised. }}
\label{fig:latency vs total throughput}
\end{figure}

\section{Conclusion}\label{conclusion}
In this paper, we proposed a unified QoS-aware framework for the next generation RAN slicing system. We aimed at maximizing long-term throughput with the MAC layer QoS constraints for hybrid legacy / 6G services under a dual scale NS and RA framework under the non-stationary channel environment. We leveraged Lyapunov optimization framework to equivalently transform the coupled backlog constrained problem into a instantaneous backlog-aware optimization problem. In stationary cases, PBRA scheme for the frame scale MINLP problem and Ad2S scheme for the super-frame long-term 
{\color{black}non-stochastic}
problem were proposed. The challenges brought by non-stationary channel variations were mitigated with the ME-KF refined scheme, e.g., Ad2S-NR. Sub-linear regret performance for both schemes were given in theoretical analysis. Numerical results showed that with our proposed QoS oriented framework and schemes, legacy and 6G applications coexisted harmonically, performance improvement in terms of QoS assurance and throughput was achieved with fewer exploration steps. 

{\color{black}
To extend the proposed unified QoS-aware framework in multi-cell scenarios, we can straightforwardly incorporate multi-agent intelligence, including multi-agent reinforcement learning \cite{zhou_multiagent_2024} or alternating direction method of multipliers \cite{yang_survey_2022}, which will be discussed in our future work.
}

\begin{appendices}
\section{Proof of Lemma~\ref{lem_dmu_long_term}}
\label{appendix:lem_1}
 As Lyapunov drift theorem is constrained under ergodic i.i.d. stochastic assumptions, we constrict \eqref{constraint average queue length, no URLLC} into super-frame based condition as given by
 
    \begin{small}
    \vspace{-6mm}
    \begin{align}\label{segment_cons}
        \frac{1}{\Bar{K}}\sum_{k = (l - 1) \times \Bar{K} + 1}^{l \times \Bar{K}} Q_n(k) \leq \left\{ 
        \begin{array}{l l}
          \delta_e, & n \in \mathcal{N}_e,\\
          \delta_M, & n \in \mathcal{N}_M, 
        \end{array} \right. 
        \forall l. 
    \end{align} 
    \vspace{-3mm}
    \end{small}
    
\noindent We already have that 
 $\Bar{K}$ is sufficiently large as presented in Section~\ref{system model}.  
 Therefore, we refer the segmented constraint \eqref{segment_cons} as a long-term constraint. By making sure that virtual queue $G_n(k)$ is stable and satisfies $\limsup_{K\to\infty}\frac{1}{K}\sum_{k = 0}^{K - 1}G_n(k) < \infty,\quad n\in \mathcal{N}_e \cup \mathcal{N}_M$, \eqref{segment_cons} can be equivalently addressed as specified in \cite{zhao_resource_2022, neely_mj_stochastic_2010}. To solve the problem of virtual queue stability, the Lyapunov-based DMU method is utilized.

The drift of Lyapunov function is

\begin{small}
\vspace{-3mm}
\begin{eqnarray}
\begin{aligned}
    &\mathbb{E}\left[L(\Gamma(k + 1)) - L(\Gamma(k))|\Gamma(k)\right] \nonumber\\
    = & \frac{\omega_Q}{2}\sum_{n\in\mathcal{N}_e \cup \mathcal{N}_M}({G_n(k + 1)}^2 - {G_n(k)}^2). \nonumber
\end{aligned}
\end{eqnarray}
\end{small}

To conserve bandwidth and power resources, $r_n(k)$ is guaranteed to be less than queue length $Q_n(k)$. Therefore, (\ref{dynamic queue model}) can be rewritten as: $Q_{n}(k+1) = Q_{n}(k) - r_n(k) + \Lambda_n(k)$. 

Based on the definition of virtual queue, we can derive

\begin{small}\vspace{-3mm}
    \begin{align}
        &{G_n(k + 1)}^2 - {G_n(k)}^2\nonumber\\
        &={\max\{G_n(k) + Q_n(k + 1) - \delta_{e,M}, 0\}}^2 - {G_n(k)}^2\nonumber\\
        &\leq[G_n(k) + Q_n(k + 1) - \delta_{e,M}]^2 - {G_n(k)}^2\nonumber\\
        &= {(Q_n(k + 1) - \delta_{e,M})}^2 + 2 G_n(k) (Q_n(k + 1) - \delta_{e,M})\nonumber\\
        &= {(Q_n(k) - \eta r_n(k) + \Lambda_n(k))}^2 + {\delta_{e,M}}^2 - 2 Q_n(k + 1) \delta_{e,M} \nonumber\\
        &+ 2 G_n(k) (Q_n(k + 1) - \delta_{e,M})\nonumber\\
        &\leq {(Q_n(k) + \Lambda_n(k))}^2 + {\delta_{e,M}}^2 + 2 G_n(k) Q_n(k) \nonumber\\
        &- 2 G_n(k) \eta r_n(k) + 2 G_n(k) \Lambda_n(k). 
    \label{drift_inequality}
    \end{align}
    \vspace{-3mm}
\end{small}

\noindent
Therefore the expectation of the Lyapunov drift with the status given at last frame can be derived as:

\begin{small}
   \begin{align}
        &\mathbb{E}\left\{{G_n(k + 1)}^2 - {G_n(k)}^2 \vert \Gamma(k)\right\} \leq C_n(k) - 2G_n(k) \eta r_n(k), \nonumber
    \end{align} 
\end{small}

\noindent $C_n(k) = {(Q_n(k) + \Lambda_n(k))}^2 + {\delta_{e,M}}^2 + 2 G_n(k) Q_n(k) + 2 G_n(k) \Lambda_n(k)$. Therefore, we have

\begin{small}
\vspace{-3mm}
    \begin{align}
        &\Delta(\Gamma(k)) - \omega_{T}\mathbb{E}\left[ \sum_{n\in\mathcal{N}_e \cup \mathcal{N}_M} r_n(k)|\Gamma(k) \right] \nonumber\\
            &\leq\mathbb{E}\left\{ \frac{\omega_Q}{2} \sum_{n\in\mathcal{N}_{e} \cup \mathcal{N}_M}\left[C_n(k) 
            - 2G_n(k) \eta r_n(k)\right] |\Gamma(k)\right\} \nonumber\\
            &- \omega_{T}\mathbb{E}\left[ \sum_{n\in\mathcal{N}_{e} \cup \mathcal{N}_U \cup \mathcal{N}_M} r_n(k)|\Gamma(k) \right]. \nonumber
    \end{align}
\vspace{-3mm}
\end{small}
\noindent

\section{Proof of Lemma~\ref{lem_bcd_converge}}
\label{appendix:lem_3_real}
Problem \ref{prob frame} can be equivalently addressed with relaxed $b_n(t, f)$ as described in Theorem~2 of \cite{zhang_network_2017} as follows

{\color{black}
\begin{small}
\vspace{-3mm}
\begin{align}
 g_{\sigma_{n^\mathrm{out}}} \triangleq \underset{\{b_n(t, f)\}, \atop \{p(t, f)\}}{\textrm{argmax}}\quad\quad& \sum_{n\in\mathcal{N}_e \cup \mathcal{N}_M} \omega_Q G_n(k) \eta r_n(k) \nonumber\\
 & + \sum_{n \in \mathcal{N}_e \cup \mathcal{N}_U \cup \mathcal{N}_M}\omega_T r_n(k)\nonumber\\
 & - \sigma^{n^\mathrm{out}} P_{\epsilon}^{n^\mathrm{out}}, \label{prob:psum}\\
 \textrm{s.t.}\qquad\quad & b_n(t,f) \in[0, 1], \quad \forall n, t, f \label{inner C1'} \tag{\ref{prob:psum}a}\\
 & \eqref{constraint sub-channel}, \eqref{constraint power}, \eqref{constraint URLLC rate}, \eqref{eqn:f_l_legacy}-\eqref{eqn:f_l}. \nonumber
\vspace{-6mm}
\end{align}
\end{small}
}

\noindent where the $p$-order (e.g., $p \in (0, 1)$) penalty term denoted as $P_{\epsilon} = \sum_{(t,f) \in \mathcal{T}_k \times \mathcal{F}} (\sum_{n \in \mathcal{N}} {(b_n(t, f) + \epsilon)}^p - c_{\epsilon, t, f})$, is added in search for binary solutions, $\epsilon > 0$ is any positive constant, $c_{\epsilon, t, f}$ follows $c_{\epsilon, t, f} = (1 + \epsilon)^p + (|\mathcal{N}| - 1) \epsilon^p$ and $\sigma > 0$ is the weighted coefficient for penalty. 

By substituting the non-concave $-P_\epsilon$ with its first order Taylor expansion lower bound $- P_{\epsilon}^{n^{\mathrm{out}}} - \nabla {P_{\epsilon}^{n^{\mathrm{out}}}}^T(\bm{b} - \bm{b}^{n^{\mathrm{out}}})$ at the $n^{\mathrm{out}}$ iteration and making $\sigma^{n^{\mathrm{out}}}$ increasing, where $\bm{b}$ is a matrix that contains all $b_n(t, f)$, we optimize the lower bound of Problem~\ref{prob frame} as  

{\color{black}
\begin{small}
\vspace{-3mm}
\begin{align}\label{convex_subchannel_allocation}
g_{n^\mathrm{out}} \triangleq \underset{\{b_n(t, f)\}, \atop \{p(t, f)\}}{\textrm{argmax}}\quad& \sum_{n\in\mathcal{N}_e \cup \mathcal{N}_M} \omega_Q G_n(k) \eta r_n(k) \nonumber\\
 & + \sum_{n \in \mathcal{N}_e \cup \mathcal{N}_U \cup \mathcal{N}_M}\omega_T r_n(k) \nonumber\\
 & - \sigma^{n^\mathrm{out}} \nabla {P_{\epsilon}^{n^{\mathrm{out}}}}^T \bm{b}, \\
 \textrm{s.t.}\qquad &\eqref{constraint sub-channel}, \eqref{constraint power}, \eqref{constraint URLLC rate}, \eqref{eqn:f_l_legacy}-\eqref{eqn:f_l}, \eqref{inner C1'}. \nonumber
\end{align} 
\end{small}
}
{\color{black}

\noindent As $g_{n^\mathrm{out}}$ is concave and constraints are convex sets for both $\{b_n(t, f)\}$ and $\{p(t, f)\}$ in \eqref{convex_subchannel_allocation}, by applying the classical block-coordinate descent (BCD) \cite{luo_delay-oriented_2017} methods to alternatively optimize $\{b_n(t, f)\}$ and $\{p(t, f)\}$, an monotonically increasing sequence of $g_{n^\mathrm{out}}(\{b_n(t, f)\}^{n^\mathrm{in}}, \{p(t, f)\}^{n^\mathrm{in}})$ is received. \cite{hong_unified_2016} prove that this block-based successive lower-bound maximization method promise the decreasing properties of original $g_{\sigma_{n^\mathrm{out}}}$, since the tight upper-bound surrogate function that majorizes the original function is minimized. Therefore, we have

\begin{small}
\vspace{-3mm}
\begin{align}
    & g_{\sigma_{n^\mathrm{out} + 1}}(\{b_n(t, f)\}^{n^\mathrm{out} + 1}, \{p(t, f)\}^{n^\mathrm{out} + 1}) \nonumber\\
    \geq & g_{\sigma_{n^\mathrm{out} + 1}}(\{b_n(t, f)\}^{n^\mathrm{out}}, \{p(t, f)\}^{n^\mathrm{out}}). \label{lower_bound_psum}
\end{align}
\end{small}

\noindent Once we denote $f_{n^\mathrm{out}} \triangleq \sum_{n\in\mathcal{N}_e \cup \mathcal{N}_M} \omega_Q G_n(k) \eta r_n(k) + \sum_{n \in \mathcal{N}_e \cup \mathcal{N}_U \cup \mathcal{N}_M}\omega_T r_n(k)$, we have

\begin{small}
\vspace{-3mm}
\begin{align}
    f_{n^\mathrm{out} + 1} -  f_{n^\mathrm{out}} \geq \sigma^{n^\mathrm{out} + 1} (P_{\epsilon}^{n^{\mathrm{out}} + 1} - P_{\epsilon}^{n^{\mathrm{out}}}). \label{inequal_psum}
\end{align}
\end{small}

\noindent Since $\sigma^{n^\mathrm{out} + 1} > 0$, $\{f_{n^\mathrm{out}}\}$ is a decreasing sequence as $\sigma^{n^{out}}$ increases, we have sequence $\{P_{\epsilon}^{n^{\mathrm{out}}}\}$ decreases. 

Suppose the positive sequence $\{\sigma^j\}$ monotonically increases and $\sigma^j \to + \infty$. Then we prove both the properties of critical point and feasibility are achieved with any limit point $\{b_n(t, f)^{j}, p(t, f)^{j}\}$ by contradiction as follows. 

\begin{itemize}
    \item {Feasibility: } assume the contrary so that $\{b_n(t, f)^{j}\}$ is not binary-valued. Then $P_{\epsilon}^{j} \neq 0$, as $\sigma^j \to + \infty$, $g_{\sigma^j} \to - \infty$, which contradicts the fact that $g_{\sigma^j}$ is lower-bounded by a limited value $g_\sigma^{j - 1} + (1 - \alpha)\sigma^{j - 1}P_{\epsilon}^{j - 1}$ as stated in \eqref{lower_bound_psum}. 
    \item {Critical point: } assume the contrary so that there exists an non-identical adjacent point $\{b_n(t, f)^{j + 1}, p(t, f)^{j + 1}\}$ other than $\{b_n(t, f)^{j}, p(t, f)^{j}\}$. Then $g_{\sigma^{j + 1}}(\{b_n(t, f)\}^{j + 1}, \{p(t, f)\}^{j + 1}) > g_{\sigma^{j + 1}}(\{b_n(t, f)\}^{j}, \{p(t, f)\}^{j})$, as $P_{\epsilon}^{j} = 0$, $f^{j + 1} - \sigma^{j + 1}P_{\epsilon}^{j + 1} > f^j$, which contradicts the fact that $\{f^j\}$ monotonically decreases. 
\end{itemize}

}

\section{Proof of Theorem~\ref{theorem_regret}}
\label{appendix:theorem_regret}

{\color{black}
Before we start the proof, let us recall Theorem 1 in \cite{neu_efficient_2020}. This theorem shows sublinear regret for LinEXP3 can be achieved under non-stochastic environments, with the hypotheses of unbiased $\theta$ estimation and linear reward model, since the reward variance is controlled via the exponential importance sampling. Then, we focus on proving that these two hypotheses stand with our proposed settings and organize the proof of Theorem~\ref{theorem_regret} as follows. 
\subsection{Proof under Stationary Cases}
\subsubsection{Linearity of Reward Model}
We prove the linear relationship between reward model $\frac{1}{\Bar{K}}\sum_{k = (l - 1) \times \Bar{K} + 1}^{l \times \Bar{K}}F^*_k(\mathcal{F}_l)$ and $\bm{X}_l$ under the closure of linear maps \cite{strang_introduction_2011}, and give the following proof. 
\begin{itemize}
    \item According to the definition of $G_n(k)$ and $C_n(k)$, both $\{G_n((l - 1) \times \Bar{K})\}$ and $\{C_n((l - 1) \times \Bar{K})\}$ are linear with $\bm{X}_l = [1, G_1((l - 1)\times\Bar{K}), Q_1((l - 1)\times\Bar{K}), Q_1^2((l - 1)\times\Bar{K}), G_1((l - 1)\times\Bar{K})Q_1((l - 1)\times\Bar{K}), \ldots, G_N((l - 1)\times\Bar{K})Q_N((l - 1)\times\Bar{K})]$. 
    \item According to the linear temporal evolution of backlog and virtual queue, when $k \in [(l - 1) \times \Bar{K} + 1, l \times \Bar{K}]$, both $\{G_n(k)\}$ and $\{C_n(k)\}$ are linear with $\{G_n((l - 1) \times \Bar{K} + 1)\}$ and $\{C_n((l - 1) \times \Bar{K} + 1)\}$. 
    \item According to the definition of $F^*_k(\mathcal{F}_l)$, for $k \in [(l - 1) \times \Bar{K} + 1, l \times \Bar{K}]$, $F^*_k(\mathcal{F}_l)$ is linear with both $\{G_n(k)\}$ and $\{C_n(k)\}$. 
\end{itemize}
By averaging $F^*_k(\mathcal{F}_l)$ through the hole super-frame $l$, the composition of up-mentioned linear functions is still linear, thus we have the linear reward model.

\subsubsection{Unbiased Estimation}
Then, we focus on proving that $\hat{\theta}_{l, f}$ is an unbiased estimator of $\theta_{l, f}$. 
}
By taking the expectation over $l$, we have

\[\scalebox{0.9}{
\begin{math}
\begin{small}
\vspace{-3mm}
\begin{aligned}
    &\mathbb{E}_l\left[ \hat{\theta}_{l, f} \right] = \mathbb{E}_l\left[ \frac{\mathbb{I}_{\{ \mathcal{F}_l = f \}}}{\pi_l(f|\bm{X}_l)} {\left({\bm{X}_l}^T {\bm{X}_l}\right)}^{-1} {\bm{X}_l} \frac{1}{\Bar{K}}\displaystyle\sum_{k = (l - 1) \times \Bar{K} + 1}^{l \times \Bar{K}}  F^*_k(\mathcal{F}_l) \right]\nonumber\\
    = &\mathbb{E}_l\left[ \mathbb{E}_l\left[\frac{\mathbb{I}_{\{ \mathcal{F}_l = f \}}}{\pi_l(f|\bm{X}_l)} \Bigg\vert \bm{X}(l)\right] \theta_{l, f} \right] = \theta_{l, f}, 
\end{aligned}
\end{small}
\end{math}
}\]

\noindent which shows that $\hat{\theta}_{l, f}$ is an unbiased estimator for $\theta_{l, f}$. 
Therefore, with the well-established reference of Theorem~1 in \cite{neu_efficient_2020}, the cumulative regret for Ad2S is bounded by \eqref{sta_regret}.  

{\color{black}
\subsection{Proof under Non-stationary Cases}
\subsubsection{Linearity of Reward Model}
We take the same path to prove the linear relationship between reward model $\frac{1}{\Bar{K}}\sum_{k = (l - 1) \times \Bar{K} + 1}^{l \times \Bar{K}}F^*_k(\mathcal{F}_l)$ and $\bm{X}_l$, and give the following proof. 
\begin{itemize}
    \item $\{r_n(k)\}$ is linear with $[1, \mathbb{E}_l[\hat{R}_1(l)], \ldots, \mathbb{E}_l[\hat{R}_N(l)]]$. Please refer to Appendix~\ref{appendix:supple} for the detailed proof. 
    \item According to the definition of $G_n(k)$ and $C_n(k)$, both $\{G_n((l - 1) \times \Bar{K})\}$ and $\{C_n((l - 1) \times \Bar{K})\}$ are linear with $[1, G_1((l - 1)\times\Bar{K}), Q_1((l - 1)\times\Bar{K}), Q_1^2((l - 1)\times\Bar{K}), G_1((l - 1)\times\Bar{K})Q_1((l - 1)\times\Bar{K}), r_1((l - 1)\times\Bar{K}), r_1^2((l - 1)\times\Bar{K}),  Q_1((l - 1)\times\Bar{K}) r_1((l - 1)\times\Bar{K}),  G_1((l - 1)\times\Bar{K}) r_1((l - 1)\times\Bar{K}), \ldots, G_N((l - 1)\times\Bar{K}) r_N((l - 1)\times\Bar{K})]$. 
    \item According to the linear temporal evolution of backlog and virtual queue, when $k \in [(l - 1) \times \Bar{K} + 1, l \times \Bar{K}]$, both $\{G_n(k)\}$ and $\{C_n(k)\}$ are linear with $\{G_n((l - 1) \times \Bar{K})\}$, $\{C_n((l - 1) \times \Bar{K})\}$ and $\{r_n(k)\}$. 
    \item According to the definition of $F^*_k(\mathcal{F}_l)$, for $k \in [(l - 1) \times \Bar{K} + 1, l \times \Bar{K}]$, $F^*_k(\mathcal{F}_l)$ is linear with $\{G_n(k)\}$, $\{C_n(k)\}$ and $\{r_n(k)\}$. 
\end{itemize}
By averaging $F^*_k(\mathcal{F}_l)$ through the hole super-frame $l$, the composition of up-mentioned linear functions is still linear, thus we have the linear reward model.

\subsubsection{Unbiased Estimation}
Then, we focus on proving that $\hat{\theta}_{l, f}$ is an unbiased estimator of $\theta_{l, f}$. 
}
By taking the expectation over $l$, we have

\[\scalebox{0.9}{
\begin{math}
\begin{small}
\vspace{-3mm}
\begin{aligned}
    &\mathbb{E}_l\left[\hat{\theta}_{l, f}\right] = \mathbb{E}_l\left[ \frac{\mathbb{I}_{\{ \mathcal{F}_l = f \}}}{\pi_l(f|\bm{X}_l)} {\left({\bm{X}_l}^T {\bm{X}_l}\right)}^{-1} {\bm{X}_l} \frac{1}{\Bar{K}}\displaystyle\sum_{k = (l - 1) \times \Bar{K} + 1}^{l \times \Bar{K}}  F^*_k(\mathcal{F}_l) \right]\nonumber\\
    = &\mathbb{E}_l\left[ \mathbb{E}_l\left[\frac{\mathbb{I}_{\{ \mathcal{F}_l = f \}}}{\pi_l(f|\bm{X}_l)} \Bigg\vert \bm{X}_l\right] \theta_{l, f} \right] = \theta_{l, f}. 
\end{aligned}  
\end{small}
\end{math}
}\]

\noindent Therefore, with the well-established reference of Theorem~1 in \cite{neu_efficient_2020}, we have the cumulative regret for Ad2S-NR bounded by \eqref{non_sta_regret}. 
{\color{black}
\section{Supplementary Proof for Appendix~\ref{appendix:theorem_regret}}
\label{appendix:supple}
With Lemma~\ref{ME-KF_convergence} and the unbiased properties of estimated $\hat{\mu}(l)$ proved in \cite{chagas_extrapolation_2016}, first we have

\[\scalebox{0.9}{
\begin{math}
\begin{small}
\vspace{-3mm}
\begin{aligned}
    & \mathbb{E}_l\left[ \hat{R}_n(l) \right] = \left\{\begin{array}{ll}
         \mathbb{E}_l\left[ 2 \hat{\mu}_n(l) \right] \cdot \log_2(e), &\hat{\mu}_n(l) \geq \frac{\ln\left( \frac{\tau |\mathcal{F}|}{P_{tot}} \right)}{2}, \\
         \mathbb{E}_l\left[ e^{-\bm{P}_{N+n, N+n}(l)} \cdot e^{2\hat{\mu}_n(l)} \right], &\hat{\mu}_n(l) < \frac{\ln\left( \frac{\tau |\mathcal{F}|}{P_{tot}} \right)}{2}, 
    \end{array}\right. \nonumber\\
    \overset{(a)}{=} & \left\{\begin{array}{ll}
         2 \mu_n(l) \cdot \log_2(e), &\hat{\mu}_n(l) \geq \frac{\ln\left( \frac{\tau |\mathcal{F}|}{P_{tot}} \right)}{2}, \\
         e^{-\bm{P}_{N+n, N+n}(l)} \cdot e^{2 \mu_n(l) + \bm{P}_{N+n, N+n}(l)}, &\hat{\mu}_n(l) < \frac{\ln\left( \frac{\tau |\mathcal{F}|}{P_{tot}} \right)}{2}, 
    \end{array}\right. \nonumber\\
    = & \left\{\begin{array}{ll}
         2 \mu_n(l) \cdot \log_2(e), &\hat{\mu}_n(l) \geq \frac{\ln\left( \frac{\tau |\mathcal{F}|}{P_{tot}} \right)}{2}, \\
         e^{2 \mu_n(l)}, &\hat{\mu}_n(l) < \frac{\ln\left( \frac{\tau |\mathcal{F}|}{P_{tot}} \right)}{2}, 
    \end{array}\right.
\end{aligned}   
\end{small}
\end{math}
}\]

\noindent $(a)$ is valid since the estimation results of Kalman filters are Gaussian distributed in linear Gaussian systems. Second, by comparing both $r_n(k), \forall k \in [(l - 1) \times \Bar{K} + 1,  l \times \Bar{K}]$ and $\hat{R}_n(l)$ under high/low SNR cases, we have

\begin{small}
\begin{align}
    & r_n(k) =  \sum_{(t, f) \in \mathcal{T}_k \times \mathcal{F}}B \log_2 \left(1 + b_n(t, f) |h_n(t,f)|^2 p(t, f) \right)\nonumber\\
    = & \sum_{(t, f) \in \mathcal{T}_k \times \mathcal{F}} B \left\{\begin{array}{ll}
         & \log_2\left(b_n(t, f) p(t, f) e^{2(\mu_n(l) + \mathcal{N}(0, \sigma_n^2))}\right), \\
         & \qquad \qquad \qquad \qquad \qquad \mu_n(l) \geq \frac{\ln\left( \frac{\tau |\mathcal{F}|}{P_{tot}} \right)}{2}, \\
         & b_n(t, f) p(t, f) e^{2\mathcal{N}(0, \sigma_n^2)} e^{2 \mu_n(l)}, \\
         & \qquad \qquad \qquad \qquad \qquad \mu_n(l) < \frac{\ln\left( \frac{\tau |\mathcal{F}|}{P_{tot}} \right)}{2}, 
    \end{array}
    \right.\nonumber\\
    = & \alpha_n(k) \mathbb{E}_l\left[ \hat{R}_n(l) \right] + \beta_n(k), 
\label{linear_representation}
\end{align}  
\end{small}

\noindent where
\[\scalebox{0.9}{
\begin{math}
\begin{small}
\vspace{-3mm}
\begin{aligned}
    &\alpha_n(k) = \sum_{(t, f) \in \mathcal{T}_k \times \mathcal{F}} \left\{\begin{array}{ll}
         B, &\mu_n(l) \geq \frac{\ln\left( \frac{\tau |\mathcal{F}|}{P_{tot}} \right)}{2}, \\
         B b_n(t, f) p(t, f)e^{2\mathcal{N}(0, \sigma_n^2)}, &\mu_n(l) < \frac{\ln\left( \frac{\tau |\mathcal{F}|}{P_{tot}} \right)}{2}, 
    \end{array}\right.
\end{aligned}   
\end{small}
\end{math}
}\]

\[\scalebox{0.9}{
\begin{math}
\begin{small}
\vspace{-3mm}
\begin{aligned}
    &\beta_n(k) = \sum_{(t, f) \in \mathcal{T}_k \times \mathcal{F}} \left\{\begin{array}{ll}
         & B \log_2\left(b_n(t, f) p(t, f)\right) + 2 \mathcal{N}(0, \sigma_n^2) \cdot \log_2(e), \\
         & \qquad \qquad \qquad \qquad \qquad \quad\mu_n(l) \geq \frac{\ln\left( \frac{\tau |\mathcal{F}|}{P_{tot}} \right)}{2}, \\
         &0, \\
         & \qquad \qquad \qquad \qquad \qquad \quad\mu_n(l) < \frac{\ln\left( \frac{\tau |\mathcal{F}|}{P_{tot}} \right)}{2}, 
    \end{array}\right.
\end{aligned}    
\end{small} 
\end{math}
}\]

\noindent are affine non-stochastic coefficients needed to be properly estimated according to our proposed scheme. To sum up, \eqref{linear_representation} indicates that any $r_n(k) \in \{r_n(k)|k \in [(l - 1) \times \Bar{K} + 1, l \times \Bar{K}]\}$ can be represented as a affine function of unbiased estimator $\hat{R}_n(l)$, which implies $\{r_n(k)\}$ is linear with $[1, \mathbb{E}_l[\hat{R}_1(l)], \ldots, \mathbb{E}_l[\hat{R}_N(l)]]$. 

\section{Performance of Our Proposed Schemes under Bursty and Varying Traffic Patterns}\label{performance_bursty}
We apply the popular bursty traffic models as specified in \cite{yang_how_2021, yang_multicast_2021} and \cite{chousainov_multiservice_2022}, and perform some extensive simulations with different inter-batch intervals $\lambda_{a, M}$ and average in-batch arrivals $\lambda_{b, M}$ as $\{1, 2, 3, 4, 5\}$ and $\{12500, 25000, 37500, 50000, 62500\}$, respectively. The corresponding simulation results are shown below, where our proposed Ad2S-NR algorithm is still promising as demonstrated in Table~\ref{tab:performance_slice_bursty} and Fig.~\ref{fig:slice_bursty}. 
\begin{table*}[t]
    \centering
    \caption{Average Slicing Configuration and MBBLL Experienced Latency under Different Arrival Profiles. As the inter-batch interval increases, the corresponding latencies for MBBLL users increase as well. However, latencies are still within 30ms threshold since our proposed Ad2S-NR dynamically allocated more spectrum resource to MBBLL slices. }
    \begin{tabular}{p{4.5cm}|c|c|c|c|c|c}
         \hline
         Inter-batch interval $\lambda_{a, M}$ & Non-bursty & 1 & 2 & 3 & 4 & 5 \\
         \hline
         \hline
         In-batch arrival $\lambda_{b, M}$ (packet/frame) & 12500 & 12500 & 25000 & 37500 & 50000 & 62500\\
         \hline
         Average slicing configuration $\Bar{|\mathcal{F}_l|}$ & 9.74 & 9.59 & 9.13 & 8.63 & 7.36 & 7.18\\
         \hline
         Average MBBLL experienced latency (ms) & 17.85 & 22.01 & 25.62 & 25.89 & 27.51 & 27.93\\
         \hline
    \end{tabular}
    \label{tab:performance_slice_bursty}
\end{table*}
\begin{figure}[H]
    \centering
    \includegraphics[width=\linewidth]{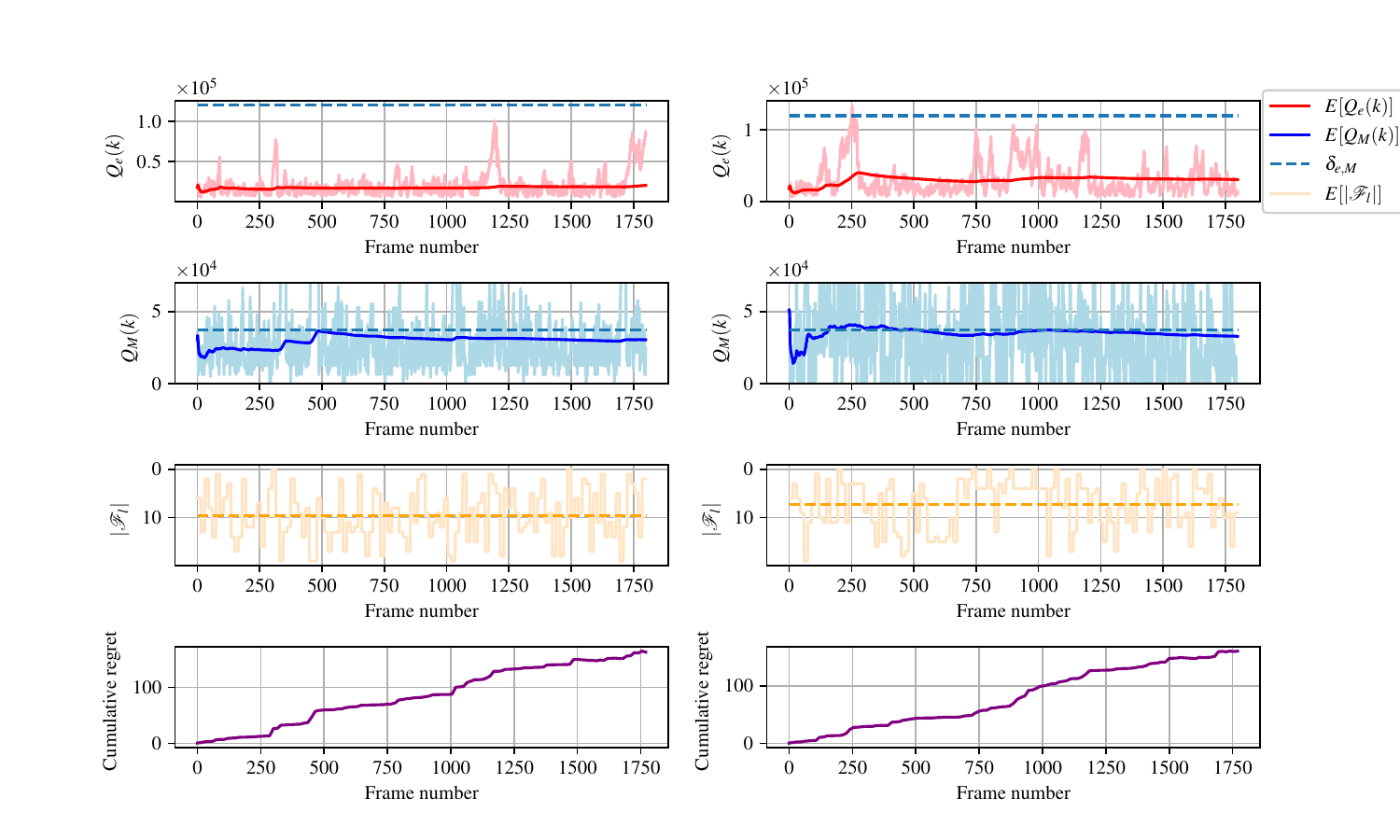}
    \caption{Slice configuration $\mathcal{F}_l$ and other performance metrics including queueing backlog and cumulative regret, where the left and right parts of this picture are carried out under $\lambda_{a, M} = 1$ and $\lambda_{a, M} = 5$, respectively. As $\lambda_{a, M}$ increases, the instantaneous queue length for MBBLL increases. To achieve QoS requirement, Ad2S-NR allocated more spectrum resource to MBBLL slices in average. }
    \label{fig:slice_bursty}
\end{figure}

We have performed extensive simulations to include the varying parameters of $\lambda_e, \lambda_U, \lambda_M$, and the corresponding numerical results are shown in Table~\ref{tab:performance_changing_lambda_varying} and Fig.~\ref{fig:latency_ptot}.
\begin{table*}[t]
    \centering
    \caption{eMBB, MBBLL Experienced Latency and URLLC satisfactory ratio under Different Varying Arrival Profiles. As the varying level in terms of $\lambda_a$ increases, average experienced latencies for both eMBB and MBBLL users increase, URLLC satisfactory ratio decreases. }
    \begin{tabular}{p{3cm}|c|c|c|c|c|c}
         \hline
         Inter-batch interval $\lambda_{a, e}, \lambda_{a, U}, \lambda_{a, M}$ & Non-bursty & 1 & 2 & 3 & 4 & 5 \\
         \hline
         \hline
         In-batch arrival $\lambda_{b, e}$ (packet/frame) & 10000 & 10000 & 20000 & 30000 & 40000 & 50000\\
         \hline
         In-batch arrival $\lambda_{b, U}$ (packet/frame) & 38 & 38 & 76 & 114 & 152 & 190\\
         \hline
         In-batch arrival $\lambda_{b, M}$ (packet/frame) & 12500 & 12500 & 25000 & 37500 & 50000 & 62500\\
         \hline
         Average eMBB experienced latency (ms) & 23.43 & 24.47 & 35.78 & 41.04 & 42.67 & 47.57\\
         \hline
         URLLC satisfactory ratio & 92.79\% & 91.55\% & 86.82\% & 84.68\% & 83.10\% & 82.12\%\\
         \hline
         Average MBBLL experienced latency (ms) & 17.85 & 22.14 & 25.87 & 27.28 & 28.50 & 29.76\\
         \hline
    \end{tabular}
    \label{tab:performance_changing_lambda_varying}
\end{table*}

\section{Performance Comparison with Reinforcement Learning-based Baseline}\label{performance_ppo}
We have conducted comparative experiments between our proposed schemes and the well-known Proximal Policy Optimization (PPO) algorithms. Experimental settings and performance comparison are shown in Table~\ref{tab:RL_setting}, \ref{tab:RL_setting_1} and Table~\ref{tab:RL_performance}, Fig.~\ref{fig:rate_vs_ptot_ppo}, \ref{fig:latency_vs_ptot_ppo}, respectively.  

\begin{table}[H]
    \centering
    \caption{PPO Experimental Hyperparameter Settings}
    \begin{tabular}{c|c}
        \toprule
        \textbf{Hyperparameter}	& \textbf{Value}\\
        \midrule
         Steps per episode & 100\\
         Discount ratio $\gamma$ & 0.99 \\
         Adam learning rate for policy optimizer & $3 \times 10^{-4}$\\
         Adam learning rate for value function optimizer & $1 \times 10^{-3}$\\
         GAE parameter $\lambda$ & 0.95\\
         Clipping ratio & 0.1\\
         \bottomrule
    \end{tabular}
    \label{tab:RL_setting}
\end{table}
\smallskip
* The source code of the adopted RL based scheme \cite{schulman_proximal_2017} can be found at \url{https://github.com/kashif/firedup}.

\begin{table}[H]
    \centering
    \caption{Experimental Hardware Settings}
    \begin{tabular}{c|c}
        \toprule
        \textbf{Parameter}	& \textbf{Value}\\
        \midrule
         CPU & Intel Core i9 \\
         Operation System & Ubuntu 22.04 \\
         Memory Capability & 64 GB \\
         Hard Disk Capability & 10 TB\\
         Number of CPUs Utilized & 1\\
         \bottomrule
    \end{tabular}
    \label{tab:RL_setting_1}
\end{table}

\begin{table*}[t]
    \centering
    \caption{Tests for Processing Delay. As shown in the table, the average processing delay of our proposed schemes is smaller than the PPO baseline, as the number of users scale up. }
    \begin{tabular}{p{4cm}|c|c|c|c|c}
        \hline
        Number of Users & 6 & 9 & 12 & 15 & 18\\
        \hline
        \hline
        Average processing delay for Ad2S-NR (s) & 0.002329 & 0.003096 & 0.004271 & 0.005176 & 0.005687\\
        \hline
        Average processing delay for PPO (s) & 0.081822 & 0.084986 & 0.084219 & 0.090774 & 0.099951\\
        \hline
    \end{tabular}
    \label{tab:RL_performance}
\end{table*}

\begin{figure}[H]
    \centering
    \includegraphics[width=\linewidth]{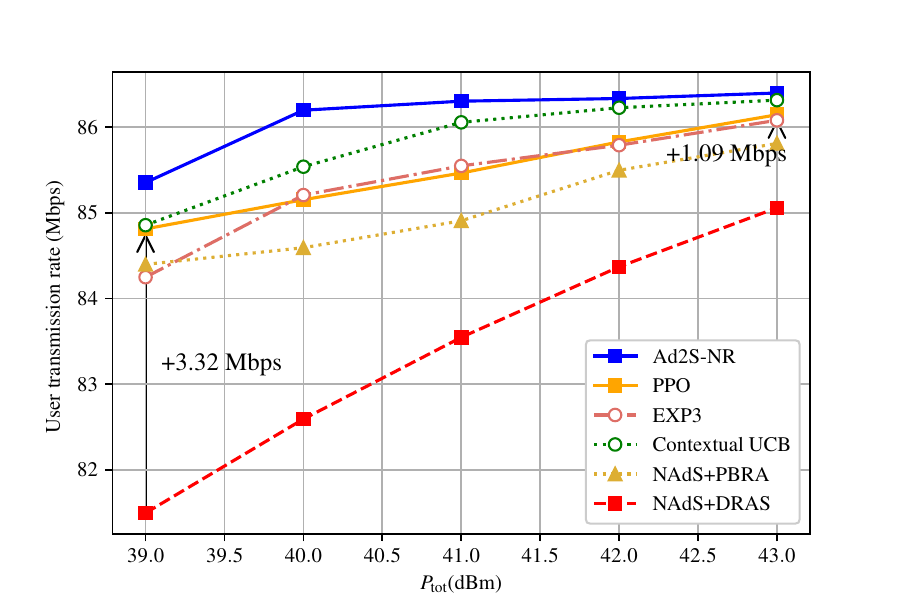}
    \caption{Average transmission rate versus total power budget. As shown in the figure, the average transmission rate of our proposed schemes is superior than PPO baseline, as the total power budget increases. }
    \label{fig:rate_vs_ptot_ppo}
\end{figure}

\begin{figure}[H]
    \centering
    \includegraphics[width=\linewidth]{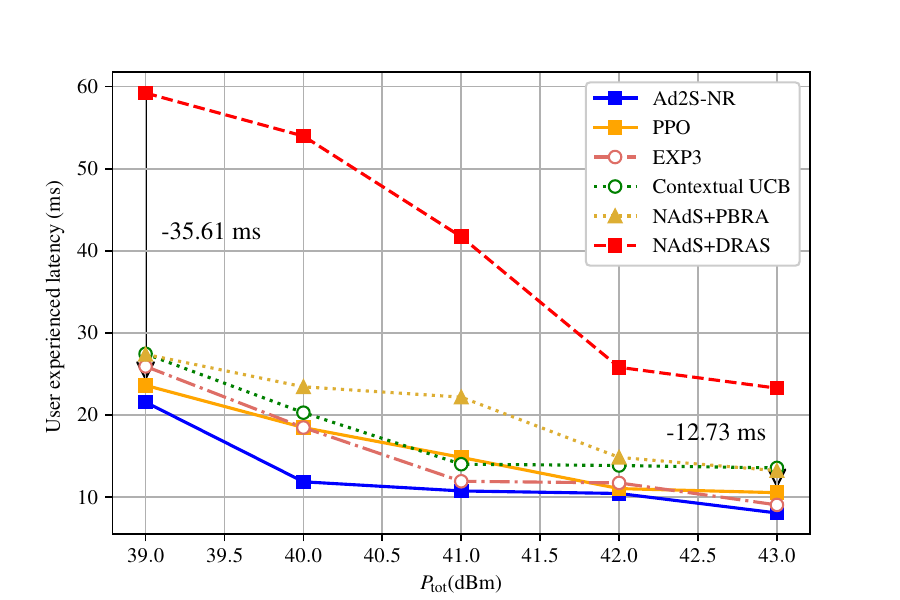}
    \caption{Average latency versus total power budget. As shown in the figure, the average latency of our proposed schemes is smaller as well, if compared to PPO baseline. }
    \label{fig:latency_vs_ptot_ppo}
\end{figure}

}
\end{appendices}

\bibliographystyle{IEEEtranc}
\bibliography{IEEEfull, iotj_acceptance_arxiv} 

\begin{thebibliography}{10}
\providecommand{\url}[1]{#1}
\csname url@samestyle\endcsname
\providecommand{\newblock}{\relax}
\providecommand{\bibinfo}[2]{#2}
\providecommand{\BIBentrySTDinterwordspacing}{\spaceskip=0pt\relax}
\providecommand{\BIBentryALTinterwordstretchfactor}{4}
\providecommand{\BIBentryALTinterwordspacing}{\spaceskip=\fontdimen2\font plus
\BIBentryALTinterwordstretchfactor\fontdimen3\font minus \fontdimen4\font\relax}
\providecommand{\BIBforeignlanguage}[2]{{%
\expandafter\ifx\csname l@#1\endcsname\relax
\typeout{** WARNING: IEEEtran.bst: No hyphenation pattern has been}%
\typeout{** loaded for the language `#1'. Using the pattern for}%
\typeout{** the default language instead.}%
\else
\language=\csname l@#1\endcsname
\fi
#2}}
\providecommand{\BIBdecl}{\relax}
\BIBdecl

\bibitem{gachet_recommendation_2023}
C.~Gachet, ``\BIBforeignlanguage{en}{Recommendation {itu}-{r} {m}.2160-0 (11/2023) - {framework} and overall objectives of the future development of {imt} for 2030 and beyond},'' Nov. 2023.

\bibitem{alwis_survey_2021}
C.~D. Alwis, A.~Kalla, Q.-V. Pham, P.~Kumar, K.~Dev, W.-J. Hwang, and M.~Liyanage, ``\BIBforeignlanguage{en}{Survey on {6G} {frontiers}: {trends}, {applications}, {requirements}, {technologies} and {future} {research}},'' \emph{\BIBforeignlanguage{en}{IEEE Open Journal of the Communications Society}}, vol.~2, p.~51, Apr. 2021.

\bibitem{elbamby_toward_2018}
\BIBentryALTinterwordspacing
M.~S. Elbamby, C.~Perfecto, M.~Bennis, and K.~Doppler, ``\BIBforeignlanguage{en}{Toward {low}-{latency} and {ultra}-{reliable} {virtual} {reality}},'' \emph{\BIBforeignlanguage{en}{IEEE Network}}, vol.~32, no.~2, pp. 78--84, Mar. 2018.
\BIBentrySTDinterwordspacing

\bibitem{3gpp_3gpp_2024}
\BIBentryALTinterwordspacing
3GPP, ``{3GPP} {work} {plan} - version {december} 6th 2024,'' Tech. Rep., Dec. 2024.
\BIBentrySTDinterwordspacing

\bibitem{zhang_joint_2018}
Y.~Zhang, S.~Bi, and Y.-J.~A. Zhang, ``Joint {spectrum} {reservation} and {on}-{demand} {request} for {mobile} {virtual} {network} {operators},'' \emph{IEEE Trans. Commun.}, vol.~66, no.~7, pp. 2966--2977, Jul. 2018.

\bibitem{anand_joint_2020}
A.~Anand, ``\BIBforeignlanguage{en}{Joint {scheduling} of {URLLC} and {eMBB} {traffic} in {5G} {wireless} {networks}},'' \emph{\BIBforeignlanguage{en}{IEEE/ACM Trans. Networking}}, vol.~28, no.~2, p.~14, Feb. 2020.

\bibitem{almekhlafi_superposition-based_2022}
\BIBentryALTinterwordspacing
M.~Almekhlafi, M.~A. Arfaoui, C.~Assi, and A.~Ghrayeb, ``\BIBforeignlanguage{en}{Superposition-{based} {URLLC} {traffic} {scheduling} in {5G} and {beyond} {wireless} {networks}},'' \emph{\BIBforeignlanguage{en}{IEEE Trans. Commun.}}, vol.~70, no.~9, pp. 6295--6309, Sep. 2022.
\BIBentrySTDinterwordspacing

\bibitem{setayesh_resource_2022}
\BIBentryALTinterwordspacing
M.~Setayesh, S.~Bahrami, and V.~W. Wong, ``\BIBforeignlanguage{en}{Resource {slicing} for {eMBB} and {URLLC} {services} in {radio} {access} {network} {using} {hierarchical} {deep} {learning}},'' \emph{\BIBforeignlanguage{en}{IEEE Trans. Wireless Commun.}}, vol.~21, no.~11, pp. 8950--8966, Nov. 2022.
\BIBentrySTDinterwordspacing

\bibitem{zhao_resource_2022}
\BIBentryALTinterwordspacing
Y.~Zhao, X.~Chi, L.~Qian, Y.~Zhu, and F.~Hou, ``\BIBforeignlanguage{en}{Resource {allocation} and {slicing} {puncture} in {cellular} {networks} {with} {eMBB} and {URLLC} {terminals} {coexistence}},'' \emph{\BIBforeignlanguage{en}{IEEE Internet Things J.}}, vol.~9, no.~19, pp. 18\,431--18\,444, Oct. 2022.
\BIBentrySTDinterwordspacing

\bibitem{chiariotti_temporal_2024}
\BIBentryALTinterwordspacing
F.~Chiariotti, M.~Drago, P.~Testolina, M.~Lecci, A.~Zanella, and M.~Zorzi, ``\BIBforeignlanguage{en}{Temporal {characterization} and {prediction} of {VR} {traffic}: {a} {network} {slicing} {use} {case}},'' \emph{\BIBforeignlanguage{en}{IEEE Trans. Mob. Comput.}}, vol.~23, no.~5, pp. 3890--3908, May 2024.
\BIBentrySTDinterwordspacing

\bibitem{tang_service_2019}
\BIBentryALTinterwordspacing
J.~Tang, B.~Shim, and T.~Q.~S. Quek, ``\BIBforeignlanguage{en}{Service {multiplexing} and {revenue} {maximization} in {sliced} {c}-{RAN} {incorporated} {with} {URLLC} and {multicast} {eMBB}},'' \emph{\BIBforeignlanguage{en}{IEEE J. Select. Areas Commun.}}, vol.~37, no.~4, pp. 881--895, Feb. 2019.
\BIBentrySTDinterwordspacing

\bibitem{yang_multicast_2021}
\BIBentryALTinterwordspacing
P.~Yang, X.~Xi, Y.~Fu, T.~Q.~S. Quek, X.~Cao, and D.~Wu, ``\BIBforeignlanguage{en}{Multicast {eMBB} and {bursty} {URLLC} {service} {multiplexing} in a {comp}-{enabled} {RAN}},'' \emph{\BIBforeignlanguage{en}{IEEE Trans. Wireless Commun.}}, vol.~20, no.~5, pp. 3061--3077, May 2021.
\BIBentrySTDinterwordspacing

\bibitem{yang_how_2021}
\BIBentryALTinterwordspacing
P.~Yang, X.~Xi, T.~Q.~S. Quek, J.~Chen, X.~Cao, and D.~Wu, ``\BIBforeignlanguage{en}{How {should} {I} {orchestrate} {resources} of {my} {slices} for {bursty} {URLLC} {service} {provision}?}'' \emph{\BIBforeignlanguage{en}{IEEE Trans. Commun.}}, vol.~69, no.~2, pp. 1134--1146, Feb. 2021.
\BIBentrySTDinterwordspacing

\bibitem{yang_ran_2021}
\BIBentryALTinterwordspacing
P.~Yang, X.~Xi, T.~Q.~S. Quek, J.~Chen, X.~Cao, and d.~Wu, ``\BIBforeignlanguage{en}{{RAN} {slicing} for {massive} {IoT} and {bursty} {URLLC} {service} {multiplexing}: {analysis} and {optimization}},'' \emph{\BIBforeignlanguage{en}{IEEE Internet Things J.}}, vol.~8, no.~18, pp. 14\,258--14\,275, Sep. 2021.
\BIBentrySTDinterwordspacing

\bibitem{wang_resource_2023-1}
L.~Wang, A.~Yin, X.~Jiang, M.~Chen, K.~Dev, N.~M. Faseeh~Qureshi, J.~Yao, and B.~Zheng, ``Resource {allocation} for {multi}-{traffic} in {cross}-{modal} {communications},'' \emph{IEEE Trans. Netw. Serv. Manage.}, vol.~20, no.~1, pp. 60--72, Mar. 2023.

\bibitem{karbalaee_motalleb_resource_2023}
\BIBentryALTinterwordspacing
M.~Karbalaee~Motalleb, V.~Shah-Mansouri, S.~Parsaeefard, and O.~L. Alcaraz~López, ``Resource {Allocation} in an {Open} {RAN} {System} {Using} {Network} {Slicing},'' \emph{IEEE Trans. Netw. Serv. Manage.}, vol.~20, no.~1, pp. 471--485, Mar. 2023.
\BIBentrySTDinterwordspacing

\bibitem{shi_two-level_2022}
\BIBentryALTinterwordspacing
W.~Shi, J.~Li, P.~Yang, Q.~Ye, W.~Zhuang, X.~Shen, and X.~Li, ``\BIBforeignlanguage{en}{Two-{level} {soft} {RAN} {slicing} for {customized} {services} in {5G}-and-{beyond} {wireless} {communications}},'' \emph{\BIBforeignlanguage{en}{IEEE Trans. Ind. Inf.}}, vol.~18, no.~6, pp. 4169--4179, Jun. 2022.
\BIBentrySTDinterwordspacing

\bibitem{kwak_dynamic_2017}
\BIBentryALTinterwordspacing
J.~Kwak, J.~Moon, H.-W. Lee, and L.~B. Le, ``\BIBforeignlanguage{en}{Dynamic network slicing and resource allocation for heterogeneous wireless services},'' in \emph{\BIBforeignlanguage{en}{2017 {IEEE} 28th {Annual} {International} {Symposium} on {Personal}, {Indoor}, and {Mobile} {Radio} {Communications} ({PIMRC})}}.\hskip 1em plus 0.5em minus 0.4em\relax Montreal, QC: IEEE, Oct. 2017, pp. 1--5.
\BIBentrySTDinterwordspacing

\bibitem{tu_deep_2024}
\BIBentryALTinterwordspacing
H.~Tu, L.~Zhao, Y.~Zhang, G.~Zheng, C.~Feng, S.~Song, and K.~Liang, ``Deep {Reinforcement} {Learning} for {Optimization} of {RAN} {Slicing} {Relying} on {Control}- and {User}-{Plane} {Separation},'' \emph{IEEE Internet Things J.}, vol.~11, no.~5, pp. 8485--8498, Mar. 2024.
\BIBentrySTDinterwordspacing

\bibitem{zhou_performance_2024}
\BIBentryALTinterwordspacing
G.~Zhou, L.~Zhao, G.~Zheng, S.~Song, and K.-C. Chen, ``Performance vs. {Cost} {Tradeoff} for {Network} {Slicing} in {Open} {RAN}: {An} {Intelligent} {Hierarchical} {Algorithm} for {Flexible} {Utility}-{Control},'' \emph{IEEE Trans. Veh. Technol.}, vol.~73, no.~11, pp. 17\,697--17\,713, Nov. 2024.
\BIBentrySTDinterwordspacing

\bibitem{mei_intelligent_2021}
\BIBentryALTinterwordspacing
J.~Mei, X.~Wang, K.~Zheng, G.~Boudreau, A.~B. Sediq, and H.~Abou-Zeid, ``\BIBforeignlanguage{en}{Intelligent {radio} {access} {network} {slicing} for {service} {provisioning} in {6G}: {a} {hierarchical} {deep} {reinforcement} {learning} {approach}},'' \emph{\BIBforeignlanguage{en}{IEEE Trans. Commun.}}, vol.~69, no.~9, pp. 6063--6078, Sep. 2021.
\BIBentrySTDinterwordspacing

\bibitem{vinogradova_estimating_2022}
\BIBentryALTinterwordspacing
J.~Vinogradova, G.~Fodor, and P.~Hammarberg, ``On {estimating} the {autoregressive} {coefficients} of {time}-{varying} {fading} {channels},'' in \emph{2022 {IEEE} 95th {Vehicular} {Technology} {Conference}: ({VTC2022}-{Spring})}, Jun. 2022, pp. 1--5, iSSN: 2577-2465.
\BIBentrySTDinterwordspacing

\bibitem{shi_unified_2021}
Q.~Shi, Y.~Liu, S.~Zhang, S.~Xu, and V.~K.~N. Lau, ``A {unified} {channel} {estimation} {framework} for {stationary} and {non}-{stationary} {fading} {environments},'' \emph{IEEE Trans. Commun.}, vol.~69, no.~7, pp. 4937--4952, Jul. 2021.

\bibitem{choi_deep_2024}
\BIBentryALTinterwordspacing
S.~Choi, S.~Choi, G.~Lee, S.-G. Yoon, and S.~Bahk, ``Deep {Reinforcement} {Learning} for {Scalable} {Dynamic} {Bandwidth} {Allocation} in {RAN} {Slicing} {With} {Highly} {Mobile} {Users},'' \emph{IEEE Trans. Veh. Technol.}, vol.~73, no.~1, pp. 576--590, Jan. 2024.
\BIBentrySTDinterwordspacing

\bibitem{yang_feeling_2022}
P.~Yang, T.~Q.~S. Quek, J.~Chen, C.~You, and X.~Cao, ``Feeling of {presence} {maximization}: {mmWave}-{enabled} {virtual} {reality} {meets} {deep} {reinforcement} {learning},'' \emph{IEEE Trans. Wireless Commun.}, vol.~21, no.~11, pp. 10\,005--10\,019, Nov. 2022.

\bibitem{zanzi_laco_2021}
\BIBentryALTinterwordspacing
L.~Zanzi, V.~Sciancalepore, A.~Garcia-Saavedra, H.~D. Schotten, and X.~Costa-Pérez, ``{laco}: {a} {latency}-{driven} {network} {slicing} {orchestration} in {beyond}-{5G} {networks},'' \emph{IEEE Trans. Wireless Commun.}, vol.~20, no.~1, pp. 667--682, Jan. 2021.
\BIBentrySTDinterwordspacing

\bibitem{zhang_learning_2024}
\BIBentryALTinterwordspacing
X.~Zhang, J.~Zuo, Z.~Huang, Z.~Zhou, X.~Chen, and C.~Joe-Wong, ``Learning {With} {Side} {Information}: {Elastic} {Multi}-{Resource} {Control} for the {Open} {RAN},'' \emph{IEEE J. Sel. Areas Commun.}, vol.~42, no.~2, pp. 295--309, Feb. 2024.
\BIBentrySTDinterwordspacing

\bibitem{chaccour_can_2022}
\BIBentryALTinterwordspacing
C.~Chaccour, M.~N. Soorki, W.~Saad, M.~Bennis, and P.~Popovski, ``\BIBforeignlanguage{en}{Can {terahertz} {provide} {high}-{rate} {reliable} {low}-{latency} {communications} for {wireless} {VR}?}'' \emph{\BIBforeignlanguage{en}{IEEE Internet Things J.}}, vol.~9, no.~12, pp. 9712--9729, Jun. 2022.
\BIBentrySTDinterwordspacing

\bibitem{zhao_online_2024}
\BIBentryALTinterwordspacing
X.~Zhao, Y.-J.~A. Zhang, M.~Wang, X.~Chen, and Y.~Li, ``\BIBforeignlanguage{en}{Online {multi}-{user} {scheduling} for {XR} {transmissions} with {hard}-{latency} {constraint}: {performance} {analysis} and {practical} {design}},'' \emph{\BIBforeignlanguage{en}{IEEE Trans. Commun.}}, pp. 1--1, 2024.
\BIBentrySTDinterwordspacing

\bibitem{bing_meta-reinforcement_2023}
\BIBentryALTinterwordspacing
Z.~Bing, D.~Lerch, K.~Huang, and A.~Knoll, ``Meta-{reinforcement} {learning} in {non}-{stationary} and {dynamic} {environments},'' \emph{IEEE Trans. Pattern Anal. Mach. Intell.}, vol.~45, no.~3, pp. 3476--3491, Mar. 2023.
\BIBentrySTDinterwordspacing

\bibitem{little_littles_2008}
\BIBentryALTinterwordspacing
J.~D.~C. Little and S.~C. Graves, ``Little's {law},'' in \emph{Building {Intuition}: {Insights} {From} {Basic} {Operations} {Management} {Models} and {Principles}}, D.~Chhajed and T.~J. Lowe, Eds.\hskip 1em plus 0.5em minus 0.4em\relax Boston, MA: Springer US, 2008, pp. 81--100.
\BIBentrySTDinterwordspacing

\bibitem{noauthor_study_2022}
\BIBentryALTinterwordspacing
``Study on scenarios and requirements for next generation access technologies,'' Apr. 2022, the 3rd Generation Partnership Project, 3GPP document 38.801.
\BIBentrySTDinterwordspacing

\bibitem{3GPP_38_401}
3GPP, ``{NG-RAN}; architecture description,'' 3GPP, Technical Specification TS 38.401, Mar. 2021, release 16.

\bibitem{3GPP_28_530}
------, ``Management and orchestration; concepts, use cases and requirements,'' 3GPP, Technical Specification TS 28.530, Jun. 2023, release 18.

\bibitem{baddour_autoregressive_2005}
\BIBentryALTinterwordspacing
K.~Baddour and N.~Beaulieu, ``Autoregressive modeling for fading channel simulation,'' \emph{IEEE Trans. Wireless Commun.}, vol.~4, no.~4, pp. 1650--1662, Jul. 2005.
\BIBentrySTDinterwordspacing

\bibitem{chopra_adaptive_2024}
\BIBentryALTinterwordspacing
R.~Chopra, C.~R. Murthy, and K.~Appaiah, ``Adaptive {data}-{aided} {time}-{varying} {channel} {tracking} for {massive} {mimo} {systems},'' \emph{IEEE Trans. Commun.}, vol.~72, no.~9, pp. 5458--5472, Sep. 2024.
\BIBentrySTDinterwordspacing

\bibitem{hamilton_time_1994}
J.~D. Hamilton, \emph{\BIBforeignlanguage{en}{Time series analysis}}.\hskip 1em plus 0.5em minus 0.4em\relax Princeton, N.J: Princeton University Press, 1994.

\bibitem{karagiannis_nonstationary_2004}
\BIBentryALTinterwordspacing
T.~Karagiannis, M.~Molle, M.~Faloutsos, and A.~Broido, ``\BIBforeignlanguage{en}{A nonstationary poisson view of internet traffic},'' in \emph{\BIBforeignlanguage{en}{{IEEE} {INFOCOM} 2004}}, vol.~3.\hskip 1em plus 0.5em minus 0.4em\relax Hong Kong, PR China: IEEE, 2004, pp. 1558--1569.
\BIBentrySTDinterwordspacing

\bibitem{boukouvala_global_2016}
\BIBentryALTinterwordspacing
F.~Boukouvala, R.~Misener, and C.~A. Floudas, ``\BIBforeignlanguage{en}{Global optimization advances in {mixed}-{integer} {nonlinear} {programming}, {minlp}, and {constrained} {derivative}-{free} {optimization}, {cdfo}},'' \emph{\BIBforeignlanguage{en}{Eur. J. Oper. Res.}}, vol. 252, no.~3, pp. 701--727, Aug. 2016.
\BIBentrySTDinterwordspacing

\bibitem{neely_mj_stochastic_2010}
M.~Neely, \emph{Stochastic {network} {optimization} with {application} to {communication} and {queueing} {systems}}, San Rafael, CA, USA: Morgan \& Claypool, 2010.

\bibitem{r_bellman_dynamic_1966}
R.~Bellman, ``Dynamic programming,'' \emph{Sci.}, vol. 153, no. 3731, pp. 4--37, 1966.

\bibitem{auer_nonstochastic_2002}
\BIBentryALTinterwordspacing
P.~Auer, N.~Cesa-Bianchi, Y.~Freund, and R.~E. Schapire, ``\BIBforeignlanguage{en}{The {nonstochastic} {multiarmed} {bandit} {problem}},'' \emph{\BIBforeignlanguage{en}{SIAM J. Comput.}}, vol.~32, no.~1, pp. 48--77, Jan. 2002.
\BIBentrySTDinterwordspacing

\bibitem{boyd_convex_2023}
S.~P. Boyd and L.~Vandenberghe, \emph{\BIBforeignlanguage{en}{Convex optimization}}, version 29~ed.\hskip 1em plus 0.5em minus 0.4em\relax Cambridge New York Melbourne New Delhi Singapore: Cambridge University Press, 2023.

\bibitem{zhang_network_2017}
\BIBentryALTinterwordspacing
N.~Zhang, Y.-F. Liu, H.~Farmanbar, T.-H. Chang, M.~Hong, and Z.-Q. Luo, ``Network {slicing} for {service}-{oriented} {networks} {under} {resource} {constraints},'' \emph{IEEE J. Select. Areas Commun.}, vol.~35, no.~11, pp. 2512--2521, Nov. 2017.
\BIBentrySTDinterwordspacing

\bibitem{neu_efficient_2020}
\BIBentryALTinterwordspacing
G.~Neu and J.~Olkhovskaya, ``\BIBforeignlanguage{en}{Efficient and robust algorithms for adversarial linear contextual bandits},'' in \emph{\BIBforeignlanguage{en}{Proceedings of {Thirty} {Third} {Conference} on {Learning} {Theory}}}.\hskip 1em plus 0.5em minus 0.4em\relax PMLR, Jul. 2020, pp. 3049--3068.
\BIBentrySTDinterwordspacing

\bibitem{anton_elementary_1987}
Anton and Howard, \emph{Elementary {linear} {algebra}}, 5th~ed.\hskip 1em plus 0.5em minus 0.4em\relax New York, NY: Wiley, 1987.

\bibitem{seldin_evaluation_2013}
\BIBentryALTinterwordspacing
Y.~Seldin, C.~Szepesvári, P.~Auer, and Y.~Abbasi-Yadkori, ``\BIBforeignlanguage{en}{Evaluation and {analysis} of the {performance} of the {EXP3} {algorithm} in {stochastic} {environments}},'' in \emph{\BIBforeignlanguage{en}{Proceedings of the {Tenth} {European} {Workshop} on {Reinforcement} {Learning}}}.\hskip 1em plus 0.5em minus 0.4em\relax PMLR, Jan. 2013, pp. 103--116, iSSN: 1938-7228.
\BIBentrySTDinterwordspacing

\bibitem{allesiardo_non-stationary_2017}
\BIBentryALTinterwordspacing
R.~Allesiardo, R.~Féraud, and O.-A. Maillard, ``\BIBforeignlanguage{en}{The non-stationary stochastic multi-armed bandit problem},'' \emph{\BIBforeignlanguage{en}{International Journal of Data Science and Analytics}}, vol.~3, no.~4, pp. 267--283, Jun. 2017.
\BIBentrySTDinterwordspacing

\bibitem{larsen_incorporation_1998}
\BIBentryALTinterwordspacing
T.~Larsen, N.~Andersen, O.~Ravn, and N.~Poulsen, ``\BIBforeignlanguage{en}{Incorporation of time delayed measurements in a discrete-time {kalman} filter},'' in \emph{\BIBforeignlanguage{en}{Proceedings of the 37th {IEEE} {Conference} on {Decision} and {Control}}}, vol.~4.\hskip 1em plus 0.5em minus 0.4em\relax Tampa, FL, USA: IEEE, 1998, pp. 3972--3977.
\BIBentrySTDinterwordspacing

\bibitem{tkailauth_linear_2000}
\BIBentryALTinterwordspacing
T.Kailauth, A.~H. Sayed, and B.~Hassibi, ``Linear {estimation},'' \emph{IEEE Trans. Inf. Theory}, vol.~51, no.~6, pp. 2236--2240, 2000.
\BIBentrySTDinterwordspacing

\bibitem{chagas_extrapolation_2016}
\BIBentryALTinterwordspacing
R.~A.~J. Chagas and J.~Waldmann, ``\BIBforeignlanguage{en}{Extrapolation of delayed measurements for fusion in a distributed sensor network},'' in \emph{\BIBforeignlanguage{en}{2016 24th {Mediterranean} {Conference} on {Control} and {Automation} ({MED})}}.\hskip 1em plus 0.5em minus 0.4em\relax Athens, Greece: IEEE, Jun. 2016, pp. 1319--1324.
\BIBentrySTDinterwordspacing

\bibitem{prabhu_sequential_2022}
\BIBentryALTinterwordspacing
G.~R. Prabhu, S.~Bhashyam, A.~Gopalan, and R.~Sundaresan, ``Sequential {multi}-{hypothesis} {testing} in {multi}-{armed} {bandit} {problems}: {an} {approach} for {asymptotic} {optimality},'' \emph{IEEE Trans. Inf. Theory}, vol.~68, no.~7, pp. 4790--4817, Jul. 2022.
\BIBentrySTDinterwordspacing

\bibitem{noauthor_study_2024}
\BIBentryALTinterwordspacing
3GPP, ``Study on channel model for frequencies from 0.5 to 100 {ghz},'' Apr. 2024, the 3rd Generation Partnership Project, 3GPP document 38.901.
\BIBentrySTDinterwordspacing

\bibitem{noauthor_system_2024}
\BIBentryALTinterwordspacing
------, ``System architecture for the {5G} {system} ({5gs}),'' Jun. 2024, the 3rd Generation Partnership Project, 3GPP document 23.501.
\BIBentrySTDinterwordspacing

\bibitem{li_contextual-bandit_2010}
\BIBentryALTinterwordspacing
L.~Li, W.~Chu, J.~Langford, and R.~E. Schapire, ``\BIBforeignlanguage{en}{A contextual-bandit approach to personalized news article recommendation},'' in \emph{\BIBforeignlanguage{en}{Proceedings of the 19th international conference on {World} wide web}}.\hskip 1em plus 0.5em minus 0.4em\relax Raleigh North Carolina USA: ACM, Apr. 2010, pp. 661--670.
\BIBentrySTDinterwordspacing

\bibitem{besbes_stochastic_2014}
\BIBentryALTinterwordspacing
O.~Besbes, Y.~Gur, and A.~Zeevi, ``Stochastic {multi}-{armed}-{bandit} {problem} with {non}-stationary {rewards},'' in \emph{Advances in {Neural} {Information} {Processing} {Systems}}, vol.~27.\hskip 1em plus 0.5em minus 0.4em\relax Curran Associates, Inc., 2014.
\BIBentrySTDinterwordspacing

\bibitem{poulkov_future_2019}
\BIBentryALTinterwordspacing
V.~Poulkov, Ed., \emph{\BIBforeignlanguage{en}{Future {Access} {Enablers} for {Ubiquitous} and {Intelligent} {Infrastructures}}}, ser. Lecture {Notes} of the {Institute} for {Computer} {Sciences}, {Social} {Informatics} and {Telecommunications} {Engineering}.\hskip 1em plus 0.5em minus 0.4em\relax Cham: Springer International Publishing, 2019, vol. 283.
\BIBentrySTDinterwordspacing

\bibitem{bracciale_lyapunov_2020}
\BIBentryALTinterwordspacing
L.~Bracciale and P.~Loreti, ``\BIBforeignlanguage{en}{Lyapunov {Drift}-{Plus}-{Penalty} {Optimization} for {Queues} {With} {Finite} {Capacity}},'' \emph{\BIBforeignlanguage{en}{IEEE Commun. Lett.}}, vol.~24, no.~11, pp. 2555--2558, Nov. 2020.
\BIBentrySTDinterwordspacing

\bibitem{zhou_multiagent_2024}
\BIBentryALTinterwordspacing
Z.~Zhou, G.~Liu, and Y.~Tang, ``Multiagent {Reinforcement} {Learning}: {Methods}, {Trustworthiness}, {Applications} in {Intelligent} {Vehicles}, and {Challenges},'' \emph{IEEE Trans. Intell. Veh.}, pp. 1--23, 2024.
\BIBentrySTDinterwordspacing

\bibitem{yang_survey_2022}
\BIBentryALTinterwordspacing
Y.~Yang, X.~Guan, Q.-S. Jia, L.~Yu, B.~Xu, and C.~J. Spanos, ``\BIBforeignlanguage{en}{A {survey} of {admm} {variants} for {distributed} {optimization}: {problems}, {algorithms} and {features}},'' Aug. 2022, arXiv:2208.03700 [cs].
\BIBentrySTDinterwordspacing

\bibitem{luo_delay-oriented_2017}
\BIBentryALTinterwordspacing
X.~Luo, ``\BIBforeignlanguage{en}{Delay-{oriented} {qos}-{aware} {user} {association} and {resource} {allocation} in {heterogeneous} {cellular} {networks}},'' \emph{\BIBforeignlanguage{en}{IEEE Trans. Wireless Commun.}}, vol.~16, no.~3, pp. 1809--1822, Mar. 2017.
\BIBentrySTDinterwordspacing

\bibitem{hong_unified_2016}
\BIBentryALTinterwordspacing
M.~Hong, M.~Razaviyayn, Z.-Q. Luo, and J.-S. Pang, ``\BIBforeignlanguage{en}{A {unified} {algorithmic} {framework} for {block}-{structured} {optimization} {involving} {big} {data}: {with} applications in machine learning and signal processing},'' \emph{\BIBforeignlanguage{en}{IEEE Signal Process Mag.}}, vol.~33, no.~1, pp. 57--77, Jan. 2016.
\BIBentrySTDinterwordspacing

\bibitem{strang_introduction_2011}
G.~Strang, \emph{\BIBforeignlanguage{en}{Introduction to {Linear} {Algebra}, {Sixth} {Edition} (2023)}}, 2011.

\bibitem{chousainov_multiservice_2022}
\BIBentryALTinterwordspacing
I.-A. Chousainov, I.~Moscholios, P.~Sarigiannidis, and M.~Logothetis, ``\BIBforeignlanguage{en}{Multiservice {loss} {models} in {C}-{RAN} {supporting} {compound} {poisson} {traffic}},'' \emph{\BIBforeignlanguage{en}{Electronics}}, vol.~11, no.~5, p. 773, Mar. 2022.
\BIBentrySTDinterwordspacing

\bibitem{schulman_proximal_2017}
\BIBentryALTinterwordspacing
J.~Schulman, F.~Wolski, P.~Dhariwal, A.~Radford, and O.~Klimov, ``\BIBforeignlanguage{en}{Proximal {policy} {optimization} {algorithms}},'' Aug. 2017, arXiv:1707.06347 [cs].
\BIBentrySTDinterwordspacing

\end{thebibliography}

\clearpage

\begin{minipage}{0.9\textwidth} 
\setlength{\parindent}{1.5em} 
\fontsize{11}{18}
\centerline{\bf Supplementary Materials}\vspace{0.2cm}

We demonstrate this supplementary results in response to the following main concerns: 

    \begin{itemize}
        \item Is the robustness of our proposed algorithms consistent under typical bursty 6G immersive applications? 
        \item Is the robustness of our proposed algorithms consistent under varying traffic patterns? 
        \item Are our algorithms competitive against deep reinforcement learning (DRL)-based baseline under non-stationary environments?  
    \end{itemize}

To address these concerns, within this material, experimental results for the performance of our proposed schemes under bursty traffic models, varying traffic and comparison to PPO baseline are included. These results are also included in Section VII-D, Appendix E and F of our main paper as demonstrated on the website: \url{https://arxiv.org/abs/2504.21444}. 
\end{minipage}
\newenvironment{singlecolumnenum}{
    \vspace{3mm}
    \begin{minipage}{0.9\textwidth}
    \begin{itemize}
}{%
    \end{itemize}%
    \end{minipage}
}
\begin{singlecolumnenum}
\vspace{3mm}
\item{\it Performance of our proposed schemes under bursty traffic models. }

We apply the popular bursty traffic models as specified in \cite{yang_how_2021, yang_multicast_2021} and \cite{chousainov_multiservice_2022}, and perform some extensive simulations with different inter-batch intervals $\lambda_{a, M}$ and average in-batch arrivals $\lambda_{b, M}$ as $\{1, 2, 3, 4, 5\}$ and $\{12500, 25000, 37500, 50000, 62500\}$, respectively. The corresponding simulation results are shown below, where our proposed Ad2S-NR algorithm is still promising. 
\begin{table}[H]
    \centering
    \begin{tabular}{p{4.5cm}|c|c|c|c|c|c}
         \hline
         Inter-batch interval $\lambda_{a, M}$ & Non-bursty & 1 & 2 & 3 & 4 & 5 \\
         \hline
         \hline
         In-batch arrival $\lambda_{b, M}$ (packet/frame) & 12500 & 12500 & 25000 & 37500 & 50000 & 62500\\
         \hline
         Average slicing configuration $\Bar{|\mathcal{F}_l|}$ & 9.74 & 9.59 & 9.13 & 8.63 & 7.36 & 7.18\\
         \hline
         Average MBBLL experienced latency (ms) & 17.85 & 22.01 & 25.62 & 25.89 & 27.51 & 27.93\\
         \hline
    \end{tabular}
    \caption{Average Slicing Configuration and MBBLL Experienced Latency under Different Arrival Profiles. As the inter-batch interval increases, the corresponding latencies for MBBLL users increase as well. However, latencies are still within 30ms threshold since our proposed Ad2S-NR dynamically allocated more spectrum resource to MBBLL slices. }
    \label{tab:performance_slice_bursty_supple}
\end{table}
\begin{figure}[H]
    \centering
    \includegraphics[width=0.8\linewidth]{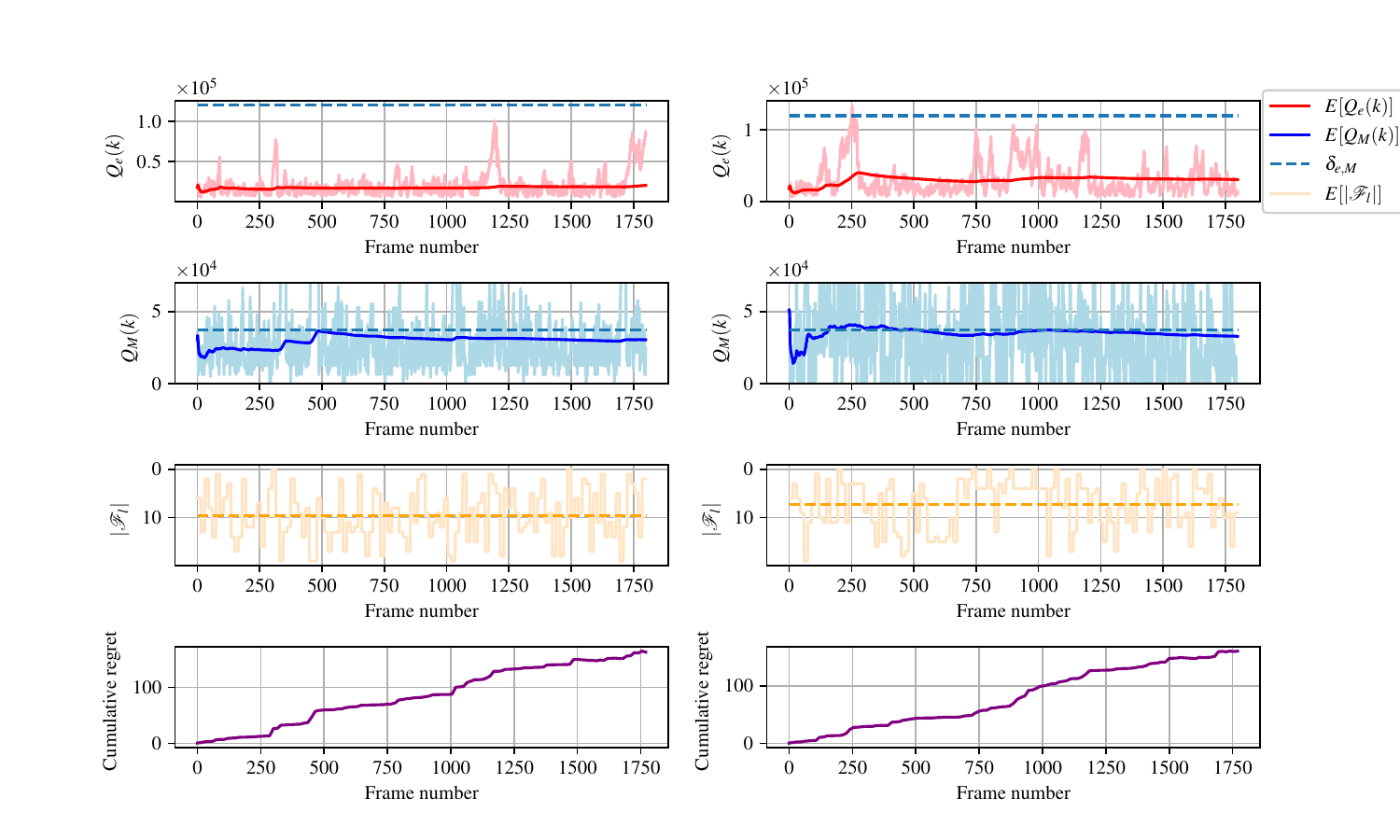}
    \caption{Slice configuration $\mathcal{F}_l$ and other performance metrics including queueing backlog and cumulative regret, where the left and right parts of this picture are carried out under $\lambda_{a, M} = 1$ and $\lambda_{a, M} = 5$, respectively. As $\lambda_{a, M}$ increases, the instantaneous queue length for MBBLL increases. To achieve QoS requirement, Ad2S-NR allocated more spectrum resource to MBBLL slices in average. }
    \label{fig:slice_bursty_supple}
\end{figure}
\end{singlecolumnenum}
\clearpage
\begin{singlecolumnenum}
\item{\it Performance of our proposed schemes under varying traffic.}

We have performed extensive simulations to include the varying parameters of $\lambda_e, \lambda_U, \lambda_M$, and the corresponding numerical results are shown below.

\begin{table}[H]
    \centering
    \begin{tabular}{|p{3cm}|c|c|c|c|c|c|}
         \hline
         Inter-batch interval $\lambda_{a, e}, \lambda_{a, U}, \lambda_{a, M}$ & Non-bursty & 1 & 2 & 3 & 4 & 5 \\
         \hline
         \hline
         In-batch arrival $\lambda_{b, e}$ (packet/frame) & 10000 & 10000 & 20000 & 30000 & 40000 & 50000\\
         \hline
         In-batch arrival $\lambda_{b, U}$ (packet/frame) & 38 & 38 & 76 & 114 & 152 & 190\\
         \hline
         In-batch arrival $\lambda_{b, M}$ (packet/frame) & 12500 & 12500 & 25000 & 37500 & 50000 & 62500\\
         \hline
         Average eMBB experienced latency (ms) & 23.43 & 24.47 & 35.78 & 41.04 & 42.67 & 47.57\\
         \hline
         URLLC satisfactory ratio & 92.79\% & 91.55\% & 86.82\% & 84.68\% & 83.10\% & 82.12\%\\
         \hline
         Average MBBLL experienced latency (ms) & 17.85 & 22.14 & 25.87 & 27.28 & 28.50 & 29.76\\
         \hline
    \end{tabular}
    \caption{eMBB, MBBLL Experienced Latency and URLLC satisfactory ratio under Different Varying Arrival Profiles. As the varying level in terms of $\lambda_a$ increases, average experienced latencies for both eMBB and MBBLL users increase, URLLC satisfactory ratio decreases. }
    \label{tab:performance_changing_lambda_supple}
\end{table}
\begin{figure}[H]
    \centering
    \includegraphics[width=0.68\linewidth]{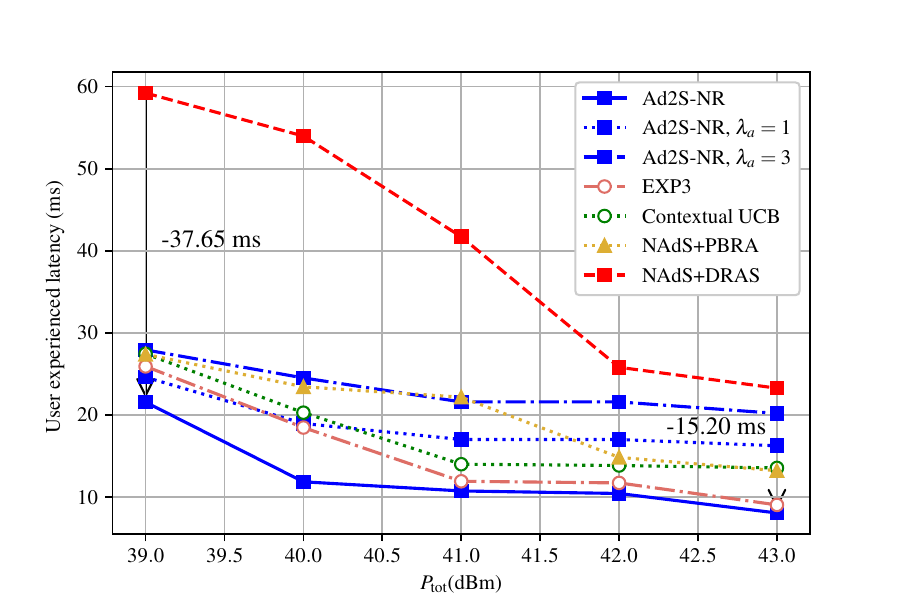}
    \caption{Average latency versus total power budget. As shown in the figure, the average latency of our proposed scheme is smaller than conventional schemes, when the power budget increases. This figure also shows that increased level of traffic variation raises difficulties and introduces extra delay. }
    \label{fig:latency_vs_ptot_bursty_supple}
\end{figure}
\end{singlecolumnenum}
\clearpage
\begin{singlecolumnenum}
\item{\it Performance comparison with PPO baseline.}

We have conducted comparative experiments between our proposed schemes and the well-known Proximal Policy Optimization (PPO) algorithms. Experimental settings and performance comparison are listed as follows. 

\begin{table}[H]
    \centering
    \caption{PPO Experimental Hyperparameter Settings}
    \begin{tabular}{c|c}
        \toprule
        \textbf{Hyperparameter}	& \textbf{Value}\\
        \midrule
         Steps per episode & 100\\
         Discount ratio $\gamma$ & 0.99 \\
         Adam learning rate for policy optimizer & $3 \times 10^{-4}$\\
         Adam learning rate for value function optimizer & $1 \times 10^{-3}$\\
         GAE parameter $\lambda$ & 0.95\\
         Clipping ratio & 0.1\\
         \bottomrule
    \end{tabular}
    \label{tab:RL_setting_supple}
\end{table}
\smallskip
* The source code of the adopted RL based scheme \cite{schulman_proximal_2017} can be found at \url{https://github.com/kashif/firedup}.

\begin{table}[H]
    \centering
    \caption{Experimental Hardware Settings}
    \begin{tabular}{c|c}
        \toprule
        \textbf{Parameter}	& \textbf{Value}\\
        \midrule
         CPU & Intel Core i9 \\
         Operation System & Ubuntu 22.04 \\
         Memory Capability & 64 GB \\
         Hard Disk Capability & 10 TB\\
         Number of CPUs Utilized & 1\\
         \bottomrule
    \end{tabular}
    \label{tab:RL_setting_supple}
\end{table}

\begin{table}[H]
    \centering
    \caption{Tests for Processing Delay. As shown in the table, the average processing delay of our proposed schemes is smaller than the PPO baseline, as the number of users scale up. }
    \begin{tabular}{|p{4cm}|c|c|c|c|c|}
        \hline
        Number of Users & 6 & 9 & 12 & 15 & 18\\
        \hline
        Average processing delay for Ad2S-NR (s) & 0.002329 & 0.003096 & 0.004271 & 0.005176 & 0.005687\\
        \hline
        Average processing delay for PPO (s) & 0.081822 & 0.084986 & 0.084219 & 0.090774 & 0.099951\\
        \hline
    \end{tabular}
    \label{tab:RL_setting_supple}
\end{table}

\begin{figure}[H]
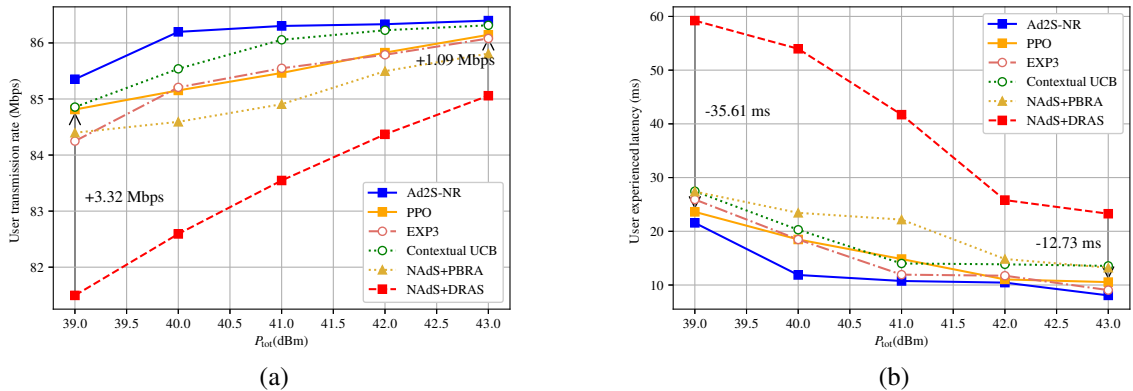

    \centering
    \subfloat[]{
        \includegraphics[width=0.48\textwidth]{fig_supplementary/rate_vs_ptot_arrow_jrnl_ppo.pdf}
        \label{fig:left}
    }
    \hfill
    \subfloat[]{
        \includegraphics[width=0.48\textwidth]{fig_supplementary/latency_vs_ptot_arrow_jrnl_ppo.pdf}
        \label{fig:right}
    }
    \caption{Average transmission rate and latency versus total power budget. As shown in the figure, the average transmission rate and latency of our proposed schemes are improved, if compared to PPO baseline. }
    \label{fig:overall}
\end{figure}

\end{singlecolumnenum}

\end{document}